\documentclass[11pt]{article}
\usepackage{graphicx}
\usepackage{amsfonts}
\usepackage{amsmath}
\usepackage{amssymb}
\usepackage{url}
\usepackage{fancyhdr}
\usepackage{indentfirst}
\usepackage{enumerate}
\usepackage{dsfont,footmisc}

\usepackage{subcaption}
\usepackage{amsthm}
\usepackage{color}
\usepackage{natbib}
\usepackage[british]{babel}
 \usepackage[colorlinks=true,citecolor=blue]{hyperref}

\usepackage[right,pagewise,displaymath, mathlines]{lineno}

\usepackage[font=scriptsize]{caption} 
\usepackage{soul,xcolor}
\setstcolor{red}

\usepackage[title]{appendix}
\def\d{\mathrm{d}}
\def\laweq{\buildrel \d \over =}

\newcommand{\var}{\mathrm{Var}}

\newcommand{\SD}{\mathrm{SD}}

\newcommand{\VaR}{\mathrm{VaR}}

\newcommand{\ES}{\mathrm{ES}}
\newcommand{\E}{\mathbb{E}}
\newcommand{\R}{\mathbb{R}}

\newcommand{\mD}{\mathcal{D}}

\newcommand{\N}{\mathbb{N}}
\newcommand{\p}{\mathbb{P}}

\newcommand{\id}{\mathds{1}}

\newcommand{\X}{\mathcal X}

\newcommand{\esssup}{\mathrm{ess\mbox{-}sup}}
\newcommand{\essinf}{\mathrm{ess\mbox{-}inf}}
\renewcommand{\ge}{\geqslant}
\renewcommand{\le}{\leqslant}
\renewcommand{\geq}{\geqslant}
\renewcommand{\leq}{\leqslant}
\renewcommand{\epsilon}{\varepsilon}

\theoremstyle{plain}
\newtheorem{theorem}{Theorem}
\newtheorem{lemma}{Lemma}
\newtheorem{proposition}{Proposition}

\theoremstyle{definition}
\newtheorem{definition}{Definition}
\newtheorem{example}{Example}
\newtheorem{assumption}{Assumption}

\theoremstyle{remark}
\newtheorem{remark}{Remark}

\theoremstyle{definition}

\renewcommand{\cite}{\citet}

\usepackage[onehalfspacing]{setspace}

\usepackage{geometry}

\begin{document}


\title{Monotonic  mean-deviation risk measures}

\author{
Xia Han\thanks{School of Mathematical Sciences and LPMC,  Nankai University, Tianjin, China.   E-mail: \texttt{xiahan@nankai.edu.cn}}
 \and Ruodu Wang\thanks{Department of Statistics and Actuarial Science, University of Waterloo, Canada. E-mail: \texttt{wang@uwaterloo.ca}}
\and  Qinyu Wu\thanks{Correspondence: Department of Statistics and Actuarial Science, University of Waterloo, Canada. E-mail: \texttt{q35wu@uwaterloo.ca}}}

 \date{\today}

\maketitle

\begin{abstract} Mean-deviation  models, along with the existing theory of coherent risk measures,  are well studied in the literature.  
In this paper,   we  characterize monotonic  mean-deviation (risk) measures from a general mean-deviation model by  applying a risk-weighting function  to the deviation part.  The  form is  a combination of  the deviation-related functional and  the expectation,  and such   measures  belong to the class of consistent risk measures.  The  monotonic  mean-deviation measures admit an axiomatic foundation  via preference relations. 
By further  assuming the   convexity  and linearity of the  risk-weighting function,  the  characterizations for  convex  and coherent  risk measures are  obtained,    giving rise to many new explicit examples of convex and nonconvex consistent risk measures. In particular, we  specialize in the convex case of the monotonic mean-deviation measure and obtain its dual representation.  
Further,  we establish asymptotic consistency and normality of the natural estimators of the monotonic  mean-deviation measures. {\color{black}Finally, the monotonic mean-deviation measures are applied to  a problem of portfolio selection  using financial data.}

\end{abstract}

\section{Introduction}

In the last few decades, risk measures  and
deviation measures  have been popular in banking and finance for various
purposes, such as calculating solvency capital reserves, pricing of insurance risks, performance
analysis, and internal risk management. Roughly speaking, deviation measures evaluate
the degree of nonconstancy in a random variable (i.e., the extent to which outcomes may deviate from a center, such as the expectation of the random variable), whereas risk measures evaluate overall prospective loss (from the benchmark of zero
loss).  
Different classes of axioms are proposed for risk measures
and deviation measures in the literature; see  \cite{ADEH99} for  coherent risk measures, \cite{FS02} and \cite{FR02} for  convex risk measures, and  \cite{RUZ06} for generalized deviation measures.   

Since the seminal work of \cite{M52}, mean-deviation or mean-risk problems are central to financial studies. In this context, a decision maker's objective functional $U$ on a loss/profit random variable $X$ can be characterized  by 
\begin{align}
\label{eq:intro}U(X)=V(\E[X],D(X)), 
\end{align}
where $\E$ is the expectation,
$V$ is a monotonic bivariate function, and $D$ measures the \emph{risk part} of $X$, which is chosen as the variance in the context of \cite{M52}, and  as a risk measure or deviation measure in subsequent studies.
For instance, the classic problem of expected return maximization with variance constraint  
  can be written as to minimize $V_{\sigma}(\E[X], \var(X))$
 where \begin{equation}\label{eq:Mark1} 
 V_\sigma(m,d):= m   + \infty \times \id_{\{d>\sigma^2 \}}\end{equation} 
 for some $\sigma >0$,\footnote{Here we interpret $X$ as loss, so  the expected return is $-\E[X]$.}
 and  it is 
 typically solved by
minimizing $V^{\lambda}(\E[X], \var(X))$, where \begin{equation}\label{eq:Mark} V^{\lambda}(m,d) := m + \lambda d\end{equation} for some $\lambda>0$ via a Lagrangian method.
Since any law-invariant coherent risk measure $R$ induces a deviation measure  $D$ via $D=R-\E$, we can write
$$
V(\E[X],R(X)) = V'(\E[X],D(X)),
$$
where $V'(m,d)=V(m,d+m)$. 
Therefore,  in this paper we focus on \eqref{eq:intro} with $D$ being a deviation measure.

The mean-deviation model is widely used in  the finance and optimization literature; see the early work of    \cite{M52}, \cite{S64},  \cite{S97}, and the more recent progresses in  \cite{GMZ12},  \cite{GZ12},   \cite{RU13},   and  \citet{HK22a,HK22b}. 
Nevertheless, there are only a few studies, including notably 
\cite{GMZ12}, that focus on the preference functional $U$  in \eqref{eq:intro}, which is an interesting mathematical object by itself, as  the decision criteria used for optimization.   

In general,   $U$  in \eqref{eq:intro} is not monotonic, as  mean-variance analysis is inconsistent with 
monotonic preferences; see, e.g.,  \cite{MMRT09}.    
Monotonicity  is self-explanatory and is common in the literature on decision theory and risk measures.   As of today, the most popular risk measures are \emph{monetary risk measures} that satisfy the  two properties of  monotonicity and cash additivity, with \emph{Value at Risk} (VaR) and \emph{Expected Shortfall} (ES) being the most famous examples. 
The monetary property allows for the interpretation of a risk measure as regulatory capital requirement defined via acceptance sets.
Therefore,  it is  natural to consider the intersection of mean-deviation models and monetary risk measures, enjoying the advantages from both streams of literature.
The functionals belonging to both classes will be called \emph{monotonic mean-deviation (risk) measures}. We omit the term ``risk" for simplicity, while keeping in mind that these functionals are risk measures in the sense of \cite{ADEH99} and \cite{FS16}.

Throughout, we consider deviation measures $D$ satisfying  the  properties of \cite{RUZ06}.  The  definitions, along with other preliminaries,  are provided in Section \ref{sec:2}.
A natural candidate for monotonic mean-deviation measures is to use the sum $U=\E+\lambda D$ for some $\lambda\ge 0$, which appears in the Markowitz model through \eqref{eq:Mark} and also in insurance pricing (see \cite{D90} and \cite{FL06}). However, this is not the only possible choice.   
  In Section \ref{sec:2new},    we characterize monotonic mean-deviation  measures among general mean-deviation models (Theorem \ref{thm:monetay}).
  It turns out that 
  they admit the form of  a combination of  the expectation and a deviation part distorted by a \emph{risk-weighting function} $g$, and $D$ needs to satisfy a condition of range normalization (defined in Section \ref{sec:2new}). Such monotonic mean-deviation measures are denoted by $\mathrm{MD}^D_g$, that is, \begin{equation}\label{eq:main-con} \mathrm{MD}^D_g=g\circ D + \E.\end{equation} 
  As far as we are aware, 
  the form of risk measures in \eqref{eq:main-con} has not been proposed in the literature, except for some special cases. 
  Although measuring both the mean and the diversification (via the deviation measure), $\mathrm{MD}^D_g$ is not necessarily a convex risk measure in the sense of \cite{FS16}. 
  Nevertheless, 
  $\mathrm{MD}^D_g$   satisfies a weaker requirement reflecting on diversification, that is, consistency with respect to second-order stochastic dominance.
 Compared with $U=\E+\lambda D$,  the risk-weighting function $g$ allows us to relax restrictions of the mean-deviation model, 
in a way similar to \cite{FS02} and \cite{FR02}, who relaxed coherent risk measures to convex ones, and to \cite{CCMTW22}, who relaxed convex risk measures to star-shaped ones. Thus, the new class of risk measures offers additional flexibility while maintaining the essential ingredients needed to assess risk via deviation in particular contexts.

 \begin{figure}[htb!]
\centering
 \includegraphics[width=15cm]{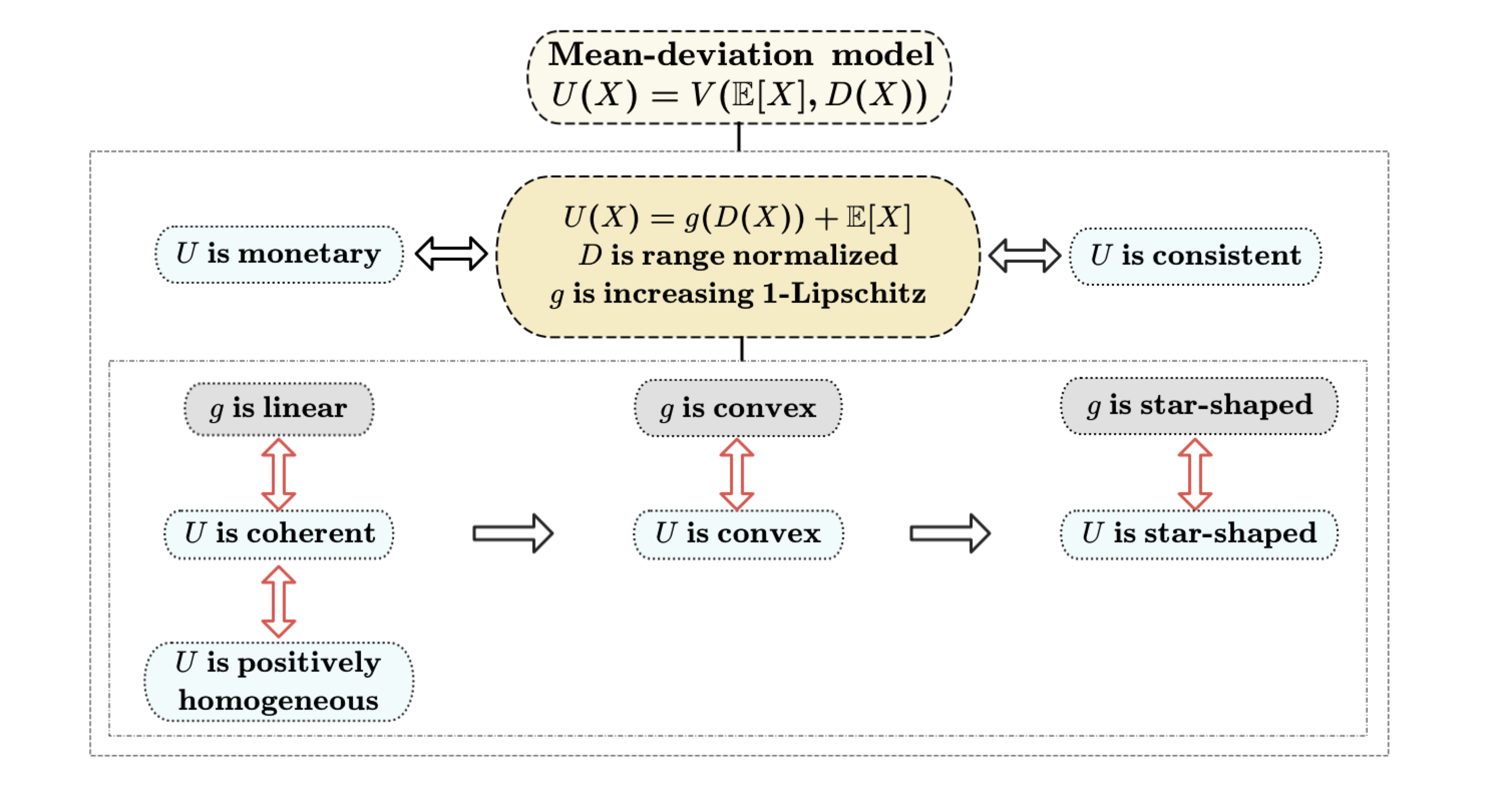}
 \captionsetup{font=small}
\caption{An illustration of properties of the mean-deviation model}\label{fig:intro}
\end{figure}

In addition to proposing the mean-deviation measures in \eqref{eq:main-con}, our main contributions include a comprehensive study on this class of risk measures. 
  In Section \ref{sec:3},    an axiomatic foundation  for $ \mathrm{MD}^D_g$ (Theorem  \ref{thm:1}) is proposed  based on results of \cite{GMZ12}, and   characterizations for coherent, convex or star-shaped  risk measures  are obtained in Theorem \ref{prop:2}.  We show that  there is a one-to-one correspondence between  $ \mathrm{MD}^D_g$  and  the risk-weighting function $g$, and hence the above classes can be identified based on properties of $g$.  
Figure \ref{fig:intro} contains an illustration of   properties of $\mathrm{MD}^D_g$.  In particular, convexity of  $g$ is equivalent to convexity  of $ \mathrm{MD}^D_g$. As a consequence,  our structure offers   new convex risk measures with explicit formulas, in addition to the existing convex distortion risk measures and entropy risk measures; see e.g., \cite{D06}, \cite{LS13} and    \cite{FS16}. 
  Specifically, 
  these formulas help us construct risk measures that are consistent yet not convex, or convex but not coherent  (Theorem \ref{prop:2} and Proposition \ref{thm:2}).

In Section \ref{sec:4},  we specialize
in the convex case of the monotonic mean-deviation measure and further  study the dual representation of  $ \mathrm{MD}^D_g$   (Theorem  \ref{th-dual}), which is  obtained directly through the conjugate function of  $g$.  
 In Section \ref{sec:6},  when the deviation measures are the  convex signed Choquet integral defined in \cite{WWW20b},   we discuss non-parametric estimation of $ \mathrm{MD}^D_g$ (Theorem \ref{thm:4}). The asymptotic normality and the asymptotic variance  for the empirical
estimators are obtained  explicitly.  
These results  yield an intuitive trade-off between statistical efficiency, in terms of estimation error,  and sensitivity to risk, in terms of the risk-weighting function.  {\color{black}In Section \ref{sec:opt}, we present an application of $ \mathrm{MD}^D_g$  in portfolio selection based on financial data, and discuss some empirical observations.}

We conclude the paper in Section \ref{sec:7}.  
The supplement regarding the characterization of monotonicity in mean-deviation models is put in Appendix \ref{app:[M]}, and
the details in the axiomatization results are relegated to  Appendix \ref{app:B}.  Appendix \ref{app:pf_thm5} provides a  proof  that is omitted from Section \ref{sec:6}. {\color{black}Appendix  \ref{sec:5}  further analyzes worst-case
values of  $ \mathrm{MD}^D_g$  under two popular settings to show its feasibility in    model uncertainty  problems.}

\section{Preliminaries
 }\label{sec:2}

Throughout this paper, we work with a nonatomic probability space $(\Omega,\mathcal F,\p)$.   Let  $\X$ be a convex cone of random variables.  All equalities and inequalities of functionals on $(\Omega,\mathcal F,\p)$ are under $\p$ almost surely ($\p$-a.s.) sense.  Let  $X\in \X$ represent the random loss faced by financial institutions in a fixed period of time. That is, a positive
value of $X\in \X$  represents a loss and a negative value represents a surplus in our sign convention, which is used by, e.g., \cite{MFE15}.     
Further, denote by $\X^\circ$ the set of all nonconstant random variables in $\mathcal X$. 
Let $F_X$ be the  distribution function of $X$, and we write $X \laweq  Y$ if two random variables $X$ and $Y$ have the same distribution. 
Terms such as increasing or decreasing functions are in the non-strict sense. 
For $p \in[1,\infty)$,  we denote by $L^p=L^p (\Omega, \mathcal{F}, \mathbb{P})$ the set  of all random variables $X$ such that $\|X\|_p=(\E[|X|^p])^{1/p}<\infty$. Furthermore, $L^\infty=L^\infty(\Omega,\mathcal F,\p)$ is the space of all essentially bounded random variables, and $ L^0=L^0(\Omega,\mathcal F,\p)$ denotes the space of all random variables. When considering a mapping defined on $L^p$ for some $p\in[1,\infty]$, we refer to its continuity in the context of the $L^p$-norm.  We denote by $\E[X]$ the expectation of the random variable $X$.
  
We  define the two important risk measures in banking and insurance practice. The Value-at-Risk ($\VaR$)  at level $\alpha\in (0,1)$ is the functional $\VaR_\alpha:L^0 \to \mathbb{R}$ defined by $$
 \VaR_\alpha(X)= \inf\{x\in \R: \p(X\le x)\ge  \alpha\},
$$
which  is precisely  the left $\alpha$-quantile of $X$. In some contexts, we also use $F^{-1}_X(\alpha)$ instead of $\VaR_\alpha(X)$ for convenience.  The Expected Shortfall ($\ES$) at  level $\alpha\in[0,1)$ is the functional $\ES_\alpha:L^1 \to \mathbb{R}$ defined by
$$
  \ES_\alpha(X)=\frac{1}{1-\alpha}\int_\alpha^1 \VaR_s(X)\d s.
$$

 \cite{ADEH99} introduced \emph{coherent risk measures} {\color{black} as functionals $\rho:\R\to(-\infty,\infty]$} that satisfy the following four properties. 
\begin{itemize}
\item[{[M]}] Monotonicity: $\rho(X)\le \rho(Y)$ for all $X,Y \in \mathcal X$ with $X\le Y$.
\item[{$[\mathrm{CA}]$}] Cash additivity:  $\rho(X+c)=\rho(X)+ c$ for all  $c\in \R$ and $X\in \X$.
\item[{$[\mathrm{PH}]$}] Positive homogeneity:  $\rho(\lambda X)=\lambda \rho(X)$ for all $\lambda \in (0, \infty)$ and $X\in \X$.
\item[{[SA]}] Subadditivity:  $\rho(X+Y)\le \rho(X)+\rho(Y)$ for all $X,Y\in \X$.
\end{itemize}
 ES satisfies all four properties above, whereas VaR does not satisfy [SA]. 
 We say that $\rho$ is a \emph{monetary risk measure}  if it satisfies [M] and [CA]. 
Moreover, $\rho$ is a \emph{convex risk measure} if  it is monetary and further satisfies 
\begin{itemize}
\item[{[Cx]}] Convexity: $\rho(\lambda X+(1-\lambda)Y)\le \lambda\rho(X)+(1-\lambda)\rho(Y)$ for all  $X,Y\in \X$ and $\lambda\in[0,1]$. 
\end{itemize} Clearly,  [PH] together with  [SA] implies [Cx]. 
Risk measures satisfying [CA] and [Cx] but not [M] are studied by e.g., \cite{FS08}.
For more discussions and interpretations of these properties, {\color{black}we refer to \cite{FS02} and \cite{FR02}}. Another class of risk measures is defined based on 
consistency with respect to second-order stochastic dominance (SSD):  
\begin{enumerate}
 \item[{[SC]}] {SSD-consistency: $\rho(X)\le \rho(Y)$ if  $X\le_{\rm SSD}Y$ (i.e.,   $\E[u(X)]\le \E[u(Y)]$ for all increasing convex functions $u$).}\footnote{Note that random variable represents the random loss instead of the random wealth. In our context, SSD is also known as increasing convex order in probability theory and stop-loss order in actuarial science. Up to a sign change converting losses to gains, SSD corresponds to increasing concave order which is the classic second-order stochastic dominance in decision theory.} 
\end{enumerate}
Monetary risk measures satisfying [SC] are called \emph{consistent risk measures}, {\color{black} introduced by  \cite{D05}}.  The consistent risk measures are characterized by \cite{MW20} as the infima of law-invariant convex risk measures (law-invariance is defined via (D5) below).
The property [SC] is often called \emph{strong risk aversion} for a preference functional in decision theory; see \cite{RS70}.
{A related notion to SSD is convex order, also called mean-preserving spread, denoted by $X\le_{\rm cx}Y$, meaning $X\le_{\rm SSD}Y$ and $\E[X]=\E[Y]$.} 
  


 In decision making,  deviation measures  
 are also introduced to measure the uncertainty inherent in a random variable, and are  studied  systematically for their application to risk management in areas like portfolio optimization and engineering.  Such measures include standard  deviation as a  special case but need not be symmetric with respect to ups and downs.  We give the definition of \emph{deviation measures}  in  \cite{RUZ06}  below.  
  \begin{definition}[Deviation measures] Fix $p\in[1,\infty]$. A \emph{deviation measure} is a functional $ D :  L^p\rightarrow [0,\infty)$ satisfying
\begin{itemize}
\item[(D1)] $ D (X+c)= D (X)$ for all $X\in L^p$ and $c\in\R$.
\item[(D2)] 
$ D (X)>0$ for  all $X\in (L^p)^\circ$. 
\item[(D3)]  $ D (\lambda X)=\lambda  D (X)$ for all $X\in L^p$ and $\lambda\ge 0$. 
\item[(D4)] $ D (X+Y) \leq  D (X)+ D (Y)$ for all $X,Y\in L^p$.
\end{itemize}\label{def:devi}
\end{definition}

Note that (D3) implies $D(0)=0$. We remark that the deviation measures in \cite{RUZ06} is defined on $L^2$ since it is easy to access tools associated with duality. However, as mentioned in \cite{RUZ06}, this does not prevent us from working with the general $L^p$ norms for $p \in[1, \infty]$.
Moreover,  we will focus on \emph{law-invariant deviation measures}, which   further satisfy \begin{itemize} 
 \item[(D5)] $D(X)=D(Y)$  for  all $X,Y\in L^p$  if  $X \stackrel{d}{=} Y$.\end{itemize}  
 We  use $\mathcal{D}^p$ to denote the set of $D$  satisfying (D1)-(D5). 
 Note that the combination of (D3) with (D4) implies that each $D\in\mathcal D^p$ is a convex functional.  
The law-invariant  deviation measures include, for instance, standard deviation,  semideviation, ES deviation and range-based deviation; see Examples 1 and 2 of  \cite{RUZ06} and Section 4.1 of \cite{GMZ12}. Moreover,  $D$ is called an  \emph{upper range dominated deviation measure}  if it has the following property   \begin{equation}\label{eq:upper}
D (X) \leq \esssup X-\E[X] ~~\text{for all~} X\in L^p,
\end{equation}  
where $\esssup X$ is the essential supremum of $X$.
For more discussions and interpretations of the properties of deviation measures mentioned above, we refer to \cite{RUZ06}.  

Deviation measures are not risk measures in the sense of  \cite{ADEH99}, but the connection between deviation measures and risk measures is strong. It is shown in Theorem 2 of  \cite{RUZ06} that   upper range bounded deviation measures $D$  correspond one-to-one with coherent, strictly expectation bounded\footnote{ A risk measure $\rho:\X \to(-\infty, \infty]$ is  \emph{strictly expectation bounded}  if it satisfies   $\rho(X)>\E[X]$ for all $X\in\X^\circ$.} risk measures $R$ with the relations  that  $D(X)=R(X)-\E[X]$ or $R(X)=D(X)+\E[X]$.  Note that the additive structure $R = D + \E$ can be seen  as a special form of the combination of mean and deviation. Below, we define a general mean-deviation model. 

\begin{definition}[Mean-deviation model] 
\label{def:MD}
Fix $p\in[1,\infty]$. For a deviation measure $D\in \mD^p$, a \emph{mean-deviation model} is a  functional  $U: L^p\rightarrow (-\infty,\infty]$  defined as 
\begin{equation}\label{eq:V}
U(X)=V(\E[X], D(X)),
\end{equation}
where $V:\R\times[0,\infty) \to (-\infty,\infty]$ satisfies
(i) $V$ is increasing component-wise; 
(ii) $V(m, 0)=m$ for all $m\in\R$; 
(iii) $V(m,d)$ is not determined only by $m$.\footnote{That is, there exist $m\in \R$ and $d_1, d_2\ge 0$  such that 
$V(m,d_1)\ne V(m,d_2)$.} 
\label{def:2}
\end{definition}

The three conditions on $V$ in  Definition \ref{def:MD} are simple and intuitive. 
More specifically, (i) is the basic requirement that $U$ increases when the mean or deviation increases, with the other argument fixed; 
(ii) means that a constant random variable has risk value equal to itself; 
(iii) means that the model is not trivial in the sense that it does not ignore  the deviation $D(X)$. 
Our definition is different from that of \cite{GMZ12}, who required further strict monotonicity of $V$
with a real-valued range. 
Therefore, our requirement is weaker than \cite{GMZ12}, and this relaxation allows us to include the most popular models of \cite{M52}  in \eqref{eq:Mark1}, that is,
$$
V(\E[X],\SD(X))= V_\sigma(\E[X],\var(X))=  \E[X]  + \infty \times \id_{\{\SD(X)>\sigma\}} ,
$$ which is neither strictly increasing nor real-valued.
Here we use $\SD$ (the standard deviation) instead of  $\var$ because $\SD\in \mathcal D^2$.

As mentioned in Introduction, the mean-deviation model has many nice properties; however,  it is not necessarily monotonic or cash additive in general,   and thus is not a monetary risk measure.  \cite{GMZ12} provided an axiomatic framework for  the mean-deviation model via the preference relation by further assuming some other properties;  but it does not belongs to the class of  monetary risk measures.   
\cite{HWWW22} characterized mean-deviation models with $D=\ES_\alpha-\E$ for $\alpha\in (0,1)$ by extending axioms for ES in \cite{WZ21}, which can be further required to be monetary.
Since [M] and  [CA]  are common in the literature of  decision theory and risk measures,
and they correspond to the interpretation of a risk measure as regulatory capital requirement, 
it is natural to further consider {general conditions for  a mean-deviation model to be  monetary.}  
This leads to the main object of this paper, monotonic mean-deviation measures, formally introduced in the next section.

\section{Monotonic mean-deviation measures}\label{sec:2new}

We first define monotonic mean-deviation measures. For this, we 
 write 
 $$
\overline{\mD}^p =  \left\{ D \in \mathcal D^p:   \sup _{X \in  {(L^p)}^\circ }\frac{D(X)}{\esssup X-\E[X]}=1\right\}.$$
Deviation measures in $\overline{\mathcal D}^p$
are called \emph{range normalized}. 
For $\lambda>0$, a real   function $g$ is \emph{$\lambda$-Lipschitz} if
 \begin{align}\label{eq-condition}
|g(x)-g(y)|\le  \lambda |x-y|~~\mbox{ for}~x,y \mbox{~in the domain of $g$}.
\end{align}
\begin{definition}
\label{def:1} Fix $p\in[1,\infty]$ and let  $D\in\overline\mD^p$. A \emph{monotonic  mean-deviation measure}  $\mathrm{MD}^D_g : {L}^p  \to \R$  is defined by   \begin{equation}\label{eq:MD}\mathrm{MD}^D_g(X)=g( D(X))+\E[X],\end{equation}
where $g:[0, \infty) \rightarrow \mathbb{R}$  is  a non-constant increasing  and  $1$-Lipschitz  function satisfying   $g(0)=0$, called   a \emph{risk-weighting function}. We use $ \mathcal{G}$ to denote the set of such functions $g$. 
\end{definition}

The interpretation of $g$ should be self-evident: it dictates  how  $D(X)$ is reflected in the calculation of  $\mathrm{MD}^D_g$, and it is a generalization of  the risk-weighting parameter $\lambda$  in $\E+\lambda D$,  hence the name.    The reason of requiring the conditions $D\in\overline\mD^p$
and $g\in \mathcal G$  in Definition \ref{def:1} will be justified in Theorem \ref{thm:monetay} below.   {\color{black}  Intuitively, since $D \in \overline{D}^p$ is not monotonic, while $\mathbb{E}$ is monotonic, a function $g$ that satisfies the 1-Lipschitz condition can mitigate the impact of $D$, ensuring the overall mean-deviation model is monotonic.  
Note that $\mathrm{MD}^D_g$ is not  subadditive in general. For instance, let  $ g(x)=(x-1)_+$ and $D=\ES_\alpha-\E$  with $\alpha=0.5$. Take  $X,Y$ be such that  $\mathbb P(X=Y=0)=\mathbb P(X=Y=2)=1/2$. 
One can check $D(X)=D(Y)=1$ and $D(X+Y)=2$. Thus, $\mathrm{MD}^D_g(X+Y)=3$ and $\mathrm{MD}^D_g(X) = \mathrm{MD}^D_g(Y)=1$, which violates subadditivity.}



\begin{theorem} \label{thm:monetay}
Fix $p\in[1,\infty]$.
Suppose that $U:L^p\to(-\infty,\infty]$ is a mean-deviation model in \eqref{eq:V} with  $D\in {\mD}^p$. The following statements are equivalent.
\begin{itemize}
\item[(i)] $U$ is a monetary risk measure.
\item[(ii)] $U$ is a consistent risk measure.
\item[(iii)] 
For some $\lambda>0$, $\lambda D\in \overline{\mD}^p $ and 
$U=g\circ D +\E $ where  $g:[0,\infty) \rightarrow \mathbb{R}$  is  a non-constant increasing and  $\lambda$-Lipschitz  function satisfying $g(0)=0$.
\end{itemize}
\end{theorem}

\begin{proof}

(ii) $\Rightarrow$ (i) is trivial.

(iii) $\Rightarrow$ (ii): 
Without loss of generality we can take $\lambda=1$. 
The property of [CA] is clear. Next, we show the property of [M]. For any $X,Y\in L^p$ with $X\le Y$, if $D(Y)\ge D(X)$, it is obviously to see that $U(X)\le U(Y)$. Assume now $D(X)>D(Y)$. It holds that
\begin{align*}
U(Y)-U(X)=g(D(Y))+\E[Y]-g(D(X))-\E[X]
\ge D(Y)-D(X)+\E[Y]-\E[X],
\end{align*}
where we have used the $1$-Lipschitz  condition of $g$ in the inequality. Since $D\in \overline{\mD}^p$, it follows from  Theorem 2 of \cite{RUZ06} that there exists one-to-one  correspondence with coherent risk measures denoted by $R$ in the relation that $R(X)=D(X)+\E[X]$  for $X\in L^p$. The {\color{black}monotonicity} of $R$ implies that
\begin{align*}
U(Y)-U(X)\ge D(Y)-D(X)+\E[Y]-\E[X]=R(Y)-R(X)\ge 0.
\end{align*}
Hence, we have verified [M] of $U$.

It remains to show that $U$ satisfies [SC]. 
Since $D$ is continuous and further satisfies convexity, law-invariance and the space is nonatomic,  we have that $D$ is consistent with respect to {convex order}  (see, e.g., Theorem 4.1 of \cite{D05}). Noting that $g$ is  increasing, we have $U$ is also consistent with respect to convex order. Combining with [M] of $U$, it follows from Theorem 4.A.6 of \cite{SS07} that $U$ satisfies [SC]. This completes the proof of (iii) $\Rightarrow$ (ii).

(i) $\Rightarrow$ (iii):
Define $g(d)=V(0,d)$ for  $d\geq0$. It is clear that   $g$ is  an  increasing function with  $g(0)=V(0,0)=0$.     By [CA], we have   $$ \begin{aligned}U (X)&=U(X-\E[X])+\E[X]\\&=V(0, D (X))+\E[X]=g( D (X))+\E[X].\end{aligned}$$
By Lemma \ref{prop:finite_K} below, we have $\lambda D\in \overline{\mathcal D}^p$ for some $\lambda>0$. 
It remains to show that $g$ is $\lambda$-Lipschitz. Denote by $k=1/\lambda$.
Since $\lambda D\in\overline \mD^p$,  for any $\epsilon\in(0,k)$, there exists $X_1$ such that $$k-\epsilon <\frac{D(X_1)}{\esssup X_1-\E[X_1]}\leq k.$$ 
For any $a>0$ and $d\ge0$, define  $$X_2=a\frac{X_1-\esssup X_1}{\esssup X_1-\E[X_1]}~~\text{and}~~X_3=\frac{d}{a}X_2+d.$$ 
It is obvious that $\E[X_2]=-a, X_2\leq 0$ and $\E[X_3]=0.$  Moreover, $a(k-\epsilon)<D(X_2)\leq ak$ and $d(k-\epsilon)<D(X_3)\leq dk$. 
Also, $\E[X_2+X_3]=-a$  and $(d+a)(k-\epsilon)<D(X_2+X_3)\leq (d+a)k$.  
Since $X_2+X_3\leq X_3$, by [M], we have $g(D(X_2+X_3))+\E[X_2+X_3]\leq g(D(X_3))+\E[X_3]$. Letting $\epsilon\to0$, we conclude that
$g_-((d+a)k)\le g(dk)+a$, 
where $g_-(x)=\lim_{y\uparrow x}g(y)$ for all $x\ge 0$. 
This is equivalent to $g_-(d+a)-g(d)\le \lambda a$ for any $a>0$ and $d\ge 0$. Note that $g$ is increasing.
We have that $g:[0,\infty)\to\R$ is  $\lambda$-Lipschitz. This completes the proof.
\end{proof}

{\color{black} Theorem \ref{thm:monetay} shows that the conditions in Definition \ref{def:1} are necessary and sufficient for \eqref{eq:MD} to be monetary, up to normalization of $D$. Specifically,
 if $D$ and $g$ satisfy the conditions in Theorem \ref{thm:monetay} (iii),
 then there exists $\widetilde D\in\overline{\mathcal D}^p$ and $\widetilde 
 g\in\mathcal G $ such that $\widetilde g\circ\widetilde D= g\circ D$.    This means that for any pair $ (D, g) $ satisfying the conditions in Theorem \ref{thm:monetay} (iii), there is a corresponding pair $ (\widetilde{D}, \widetilde{g}) $ that  satisfies the conditions  in Definition \ref{def:1}.
Therefore, Definition \ref{def:1} includes all possible choices of mean-deviation models satisfying [M] and [CA].}

The following result is needed in the proof of Theorem \ref{thm:monetay}. 

\begin{lemma}\label{prop:finite_K} 
Fix $p\in[1,\infty]$, and let $D\in \mathcal D^p$.
If $ U=V(\E,D)$  in \eqref{eq:V} satisfies [M], then we have
$U(X)<\infty$ for all $X\in L^p$, and
there exists $\lambda>0$ such that $\lambda D \in \overline{\mD}^p$. 
\end{lemma}
\begin{proof} 
To show that $U(X)$ is finite, take $Y\in L^\infty$ such that $\E[Y]=\E[X]$ and $D(Y)=D(X)$. Such $Y$ exists because $D$ is positively homogeneous. Therefore, $U(X)=U(Y)\le U(\esssup Y)=V(\esssup Y,0)=\esssup Y <\infty.$

We next prove 
 \begin{equation}\label{eq:K}
  K:= \sup _{X \in  {(L^p)}^\circ }\frac{D(X)}{\esssup X-\E[X]}<\infty.\end{equation}
 by contradiction, which 
is equivalent to $\lambda D \in \overline{\mD}^p $ with $\lambda=1/K$. 
Assume that $K=\infty$ in \eqref{eq:K}. For $X_1, X_2\in L^p$ such that $\E[X_1]<\E[X_2]$,  let $m_1=\E[X_1]$, $d_1=D(X_1)$, $m_2=\E[X_2]$, $d_2=D(X_2)$, and $e=m_2-m_1$. If $K=\infty$, there exists  $Y_1$ such that
$
{D\left(Y_1\right)}/{(\esssup Y_1-\E[Y_1])} \geq {d_1}/{e} .
$
Denote by
$
Y_2=e {(Y_1-\esssup Y_1)}/{(\esssup Y_1-\E[Y_1])}+m_2.
$
It holds that $\E\left[Y_2\right]=-e+m_2=m_1$ and ${D}\left(Y_2\right) \geq d_1$, and  thus  $U(X_1)=V( m_1,d_1) \leq V(\E[Y_2],D(Y_2))=U(Y_2)$. On the other hand, observe that $Y_2 \leq m_2$. Consequently, by  monotonicity,  we have $U(Y_2)\leq U(m_2)$.  Thus, we have $U(X_1)\leq U(Y_2)\leq U(m_2)\leq U(X_2),$
which implies that $U(X) \leq U(Y)$ for every $X$ and $Y$ with $\E[X]<\E[Y]$.  
Hence, $$\E[X]-\epsilon=U(\E[X]-\epsilon)\le U(X)\le U(\E[X]+\epsilon)=\E[X]+\epsilon$$ for any $X\in\ L^p$ and $\epsilon>0$.
Letting $\epsilon \downarrow 0$ yields $U(X)=\E[X]$, contradicting (iii) in Definition \ref{def:MD}.   Therefore, we conclude $K<\infty.$ This completes the proof.
\end{proof}

Lemma \ref{prop:finite_K} provides a necessary condition for [M] on the mean-deviation model. For the sake of completeness, we will elaborate on the characterization of [M], under the assumption that this necessary condition is met; this is detailed in Appendix \ref{app:[M]}. Lemma \ref{prop:finite_K} also implies in particular that we can limit the range of the mean-deviation model to $\R$ when [M] is imposed.

For a deviation measure defined on $L^p$, the condition that $\lambda D\in\overline{\mathcal D}^p$ for some $\lambda>0$ in Lemma \ref{prop:finite_K} is called  weakly upper-range dominated by \cite{GMZ12}.
It is clear from  \eqref{eq:upper} that  every upper-range dominated deviation measure is  weakly upper-range dominated with $\lambda\ge 1$.   In particular,  if  $D$ takes the form of   $\ES_\alpha-\E$ with $\alpha\in(0,1)$,   $\esssup X-\E [X]$ or $\E[|X-\E[X]|]/2$, we have $D\in \overline \mD^p$ (see Example 5 of \cite{GMZ12} for the last one).
Examples of deviation measures $D$ satisfying  $\lambda D \in \overline{\mathcal D}^p$ for some $\lambda>0$ also include the mean-absolute deviation, the Gini deviation, the inter-ES range, and the inter-expectile range (for the last two, see  \cite{BFWW22}).



Note that $R=D+\E$ is a finite coherent risk measure on $L^p$
for any $D\in\overline{\mathcal D}^p$. It follows that $R$ is continuous  (see e.g., Corollary 2.3 of \cite{KR09}). Consequently, this implies that a range-normalized deviation measure is always continuous.
Below we characterize the class of range-normalized deviation measures by elucidating the relationship between coherent risk measures and the deviation measures in $\overline{D}^p$.
\begin{proposition}
\label{prop:verify}
Fix $p\in[1,\infty]$. The deviation measure $D\in \mathcal D^p$ is range normalized  if and only if $ D+\E$ is a coherent risk measure and $ \lambda D+\E$ is not a coherent risk measure for $\lambda >1$.
\end{proposition}
\begin{proof}
The necessity follows immediately from Theorem 2 of \cite{RUZ06} since $D\in\overline{\mathcal D}^p$ is upper range dominated (see \eqref{eq:upper} for the definition) and $\lambda D$ is not upper range dominated for any $\lambda>1$. Conversely, we assume by contradiction that $D$ is not range normalized. Then, either $k D\in \overline{\mathcal D}^p$ for some $k>1$ or $k D\in \overline{\mathcal D}^p$ for some $k<1$ holds.


In the first case, there exists $\lambda>1$ such that $\lambda D$ is upper range dominated. Applying Theorem 2 of \cite{RUZ06}, we have
$\lambda D+\E$ is a coherent risk measure, thereby leading to a contradiction. In the second case, it holds that $D+\E$ is not a coherent risk measure since $D$ is not upper range dominated, which also yields a contradiction. 
\end{proof}


{The expected return maximization with variance constraint of \cite{M52} has the form
$
\mathrm{MD}_g^D
$ 
where $D=\SD$ and $g(d)=\infty \times \id_{\{d>\sigma\}}$ for some $\sigma>0$ as in \eqref{eq:Mark1}.
In this example, $g$ is not real-valued. 
Therefore, although sharing the form \eqref{eq:MD}, 
$
\mathrm{MD}_g^D
$ is not a monotonic mean-deviation measure. 
Similarly, 
for $\lambda>0$, the functional
$\mathrm{MD}_g^D(X)= \lambda (\SD(X))^2+\E[X]$ 
in \eqref{eq:Mark}
or  
$\mathrm{MD}_g^D(X)= \lambda \SD(X)+\E[X]$  is not a monotonic mean-deviation measure, because $\SD$
does not satisfies \eqref{eq:K}
for any $p\in [1,\infty]$. 
Nevertheless, in all three examples, $g$ is convex. Indeed, convexity of $g$ has important implications, and this will be studied in Section \ref{sec:3.2} below.
}



 \section{Characterization}\label{sec:3}

\subsection{Axiomatization of monotonic mean-deviation measures}\label{sec:3.1}

We first present an axiomatization of the monotonic mean-deviation measure $\mathrm{MD}^D_g$ through preference relations. This axiomatization is very similar to  \cite{GMZ12}, who axiomatized  preferences represented by  a monotone functional $X\mapsto V(\E[X],D(X))$ for some strictly increasing function $V$.
We  relegate all details, including all proofs and a comparison with \cite{GMZ12},  to Appendix \ref{app:B}.
Our main purpose here is to show that $\mathrm{MD}^D_g$ has an axiomatic foundation. 

  A preference relation  $\succeq$ is defined by  a total preorder.\footnote{A preorder is a binary relation  on $\mathcal{X}$, which is reflexive and transitive. A binary relation $\succeq$ is reflexive if $X \succeq X$ for all $X \in \mathcal{X}$, and transitive if $X \succeq Y$ and $Y \succeq Z$ imply $X \succeq Z$. A  total preorder is a preorder which in addition is complete, that is, $X \succeq Y$ or $Y \succeq X$ for all $X, Y \in \mathcal{X}$.\label{foot3}} 
 As usual,   $\succ$ and $ \simeq $  correspond to the antisymmetric and equivalence relations, respectively. For two random losses $X,Y$, the relation  $X\succeq Y$  indicates that 
$X$ is preferred over $Y$, or equivalently, that 
$Y$ is considered more dangerous than $X$.  A numerical representation of a preference  $\succeq$  is a mapping $\rho: \X \rightarrow\R$, such that  $X \succeq Y \Longleftrightarrow \rho(X)\le \rho(Y)$. Note that  $\succeq$ can be represented by a mapping  $\rho$  if  $\succeq$ is separable;  see e.g.,   \cite{DK13}.\footnote{A total preorder $\succeq$ is separable if there exists a countable set $\X \subseteq L^p $  for  $p\in[1,\infty]$ such that for any $x, y \in \X$ with $x\succ  y$ there is $z \in \X$ for which $x \succeq z \succeq y$.} 
     We use the following  axioms,  where all random variables are tacitly assumed to be in $ L^p$ for some fixed $p\in[1,\infty]$.  
\begin{itemize}
\item[{A1}] (Monotonicity).  If $X_1\le X_2,$ then $X_1\succeq X_2$.

\item[A2] (Translation-invariance).  For any $c \in \R$,   $X\succeq Y$ if and ony if $X+c\succeq Y+c$. 
\item[A3] (Weak positive homogeneity).  
{If 
 $\E[X]=\E[Y]$ and $X\succeq Y$, then $ \lambda X\succeq \lambda Y$ for any $\lambda >0$. }
\item[A4](Risk aversion). 
{If $X\le_{\rm cx}Y$,
 then  $X \succeq  Y$.
In addition, 
$\E[X]\succ X$ for any non-constant $X$.}
\item[A5] ({Solvability}). There exists $c \in \R$ such that $X \simeq c$. 
\item[A6]  (Weak convexity). If $\E[X]=\E[Y] $ and $X\simeq Y$, then $ \lambda X+(1-\lambda) Y \succeq X$ for all $\lambda \in[0,1]$.
\item[A7] ({Continuity}). {For every $X$,} the sets  $\left\{Y \in L^p:  Y \succeq X\right\}$  and
$ \left\{Y \in L^p : X \succeq Y\right\}$ 
are  $L^p$-closed.
\end{itemize}

These axioms are standard, and we refer to \cite{Y87}, \cite{DK13},  \citet[Chapeter 2]{FS16} and \cite{GMZ12}   for interpretations and discussions of these axioms. 
The following result gives an axiomatization of $\mathrm{MD}^D_g$ in Definition \ref{def:1}.
  \begin{theorem}\label{thm:1}
{Fix $p\in [1,\infty]$.} A preference  $\succeq$ on $L^p$ satisfies Axioms A1--A7  if and only if  $\succeq$ can be represented by $\mathrm{MD}^D_g=g \circ D+\mathbb{E}$ for some  $D\in \overline \mD^p$  and  {$g\in \mathcal G$ that is strictly increasing.}
\end{theorem} 
 
\cite{GMZ12} obtained a representation  with the form $X\mapsto V(\E[X],D(X))$  using a weak translation-invariant property, and  the mean-deviation model  is not cash-additive.   
Our stronger version of translation-invariance pins down 
the more explicit form of  monotonic mean-deviation  measures. {\color{black} We will also establish   explicit  one-to-one correspondence between    properties of the risk measure  $\mathrm{MD}^D_g$  and   properties of the risk-weighting function $g$ in the next section. }

 For the detailed differences between our axiomatization and that of \cite{GMZ12}, see Appendix \ref{app:B}. A subtle difference between Theorem \ref{thm:1} and Definition \ref{def:1} is that  $g$ is strictly increasing in  Theorem \ref{thm:1} but not necessarily so in  Definition \ref{def:1}.  An axiomatization of $\mathrm{MD}^D_g$ with $g$ not necessarily strictly increasing is an open question, as we were not able to identify proper relaxations of the proposed axioms.


 
\subsection{Characterizations of convex and coherent risk measures}\label{sec:3.2}
In this subsection, we continue to study   properties of  $\mathrm{MD}^D_g$.     Specifically, we characterize  $g$  such that   $\mathrm{MD}^D_g$  belongs  to the class of   coherent risk measures or convex  risk measures.  
{Moreover, we consider \emph{star-shaped risk measures}, which are monetary risk measures $\rho$ further satisfying 
\begin{itemize}
\item[{[SS]}] Star-shapedness: $\rho(0)=0$ and $\rho(\lambda X)\le \lambda\rho(X) $ for all  $X \in \X$ and $\lambda\in[0,1]$. 
\end{itemize}
Similarly, a function $g :[0,\infty)\to \R$ is \emph{star-shaped} if $g(0)=0$ and $g(\lambda x)\le \lambda g(x) $ for all  $x \in [0,\infty)$ and $\lambda\in[0,1]$. 
{Star-shaped risk measures} are characterized by \cite{CCMTW22} as the infimum of normalized (i.e., $\rho(0)=0$) convex risk measures.
Under normalization,   star-shapedness is weaker than 
both convexity and  positive homogeneity. {\color{black}We refer to  \cite{HK22b},  \cite{LGZ24} and \cite{NTJ24}  for more recent developments on star-shaped risk measures.} }

\begin{theorem}\label{prop:2} 
Suppose that  $D\in \overline \mD^p$ for   $p\in[1,\infty]$ and $g\in\mathcal G$.  The following statements hold. 
\begin{enumerate}
\item[(i)]  $\mathrm{MD}^D_g$ is a coherent risk measure if and only if $g$ is linear.
\item[(ii)]  $\mathrm{MD}^D_g$ is a convex risk measure   if  and only if $g $   is convex.  
\item[(iii)] {$\mathrm{MD}^D_g$ is a star-shaped risk measure if and only if  $g$ is star-shaped.}
\end{enumerate}
\end{theorem}
\begin{proof}  
Sufficiency is straightforward. To show necessity,    let  $X$ be such that  $\E[X]=0$ and $D(X)=1$; such $X$ exists due to Property (D3). Coherence of $\mathrm{MD}^D_g$ implies that for all $x>0$,
$$
g(x )=\mathrm{MD}^D_g(x X)=x \mathrm{MD}^D_g(X)= x  g(1).
$$
This implies that $g$ is linear.

 (ii) To see sufficiency, if  $g$ is convex, then   $\mathrm{MD}^D_g$ is a convex risk measure because expectation is linear and $D$ is  convex.  {To show necessity,  
take $x,y\ge0$ and $\lambda \in [0,1]$. Let  $X$ be such that  $\E[X]=0$ and $D(X)=1$.
Since $\mathrm{MD}^D_g$ is convex and $D$ satisfies (D3), we have 
\begin{align*}
g(\lambda x+(1-\lambda)y)
&= g \circ D((\lambda x+ (1-\lambda)y)X
)
\\ &=\mathrm{MD}^D_g((\lambda x+ (1-\lambda)y)X )
\\&\le \lambda\mathrm{MD}^D_g(x X)+(1-\lambda)\mathrm{MD}^D_g(y X)=\lambda g(x) + (1-\lambda)g(y).
\end{align*}
Thus, $g$ is convex.}

 (iii)
 To see sufficiency, if $g$ is star-shaped, then ${\rm MD}_g^D$ is star-shaped because expectation is linear and $D$ satisfies (D3). Conversely, let $X$ be such that  $\E[X]=0$ and $D(X)=1$. For any $x\in[0,\infty)$ and $\lambda\in[0,1]$, it follows from the star-shapedness of ${\rm MD}_g^D$ that $g(0)={\rm MD}_g^D(0)=0$ and
 \begin{align*}
g(\lambda x)={\rm MD}_g^D(\lambda x X)\le \lambda{\rm MD}_g^D( x X)=\lambda g(x).
 \end{align*}
This implies that $g$ is star-shaped. 
\end{proof}

 By Theorem \ref{prop:2} (i),   $\mathrm{MD}^D_g$ is coherent if and only if  $$\mathrm{MD}^D_g(X)=
\lambda D (X) + \E[X] =\lambda  R(X)+(1-\lambda )\mathbb{E}[X],~~X\in L^p$$ for some $\lambda \in [0,1]$, where  $R=D+\E$ is a coherent risk measure.  
In fact,
positive homogeneity of $\mathrm{MD}^D_g$ is sufficient for  $g$ to be linear, as seen from the proof of (i). Therefore, for $\mathrm{MD}^D_g$, positive homogeneity and coherence are equivalent.
Moreover, following the same proof, the result in (ii) can be strengthened to a more general form without monotonicity: For any function $g:[0,\infty)\to \R$ and $D\in\mathcal D^p$ with $p\in[1,\infty]$, we have that $\mathrm{MD}^D_g$ is  convex    if  and only if $g $   is convex. 
{\color{black} As shown in   Proposition 3 of \cite{CCMTW22},   if $\mathrm{MD}^D_g$  is  subadditive, then the three (coherent, convex, star-shaped) classes of risk measures in   Theorem \ref{prop:2} (i)--(iii) coincide. 
}

{For the special choice of  $D=\ES_\alpha-\E$ where $\alpha\in(0,1)$,  \cite{HWWW22} obtained characterizations for $\mathrm{MD}^D_g$ to be coherent, convex, or consistent risk measures.   Theorem \ref{prop:2} extends this result to deviation measures. 
Our results allow for explicit formulas for many consistent risk measures that are not convex. In contrast, existing examples of consistent but non-convex risk measures are often obtained by taking an infimum over convex risk measures.  
 }
 
In the following proposition, we obtain an alternative representation result for $\mathrm{MD}^D_g$ when $g$ is convex.  





\begin{proposition}\label{thm:2} Fix $p\in[1,\infty]$.  For $g\in\mathcal G$ and $D\in \overline {\mathcal D}^p$,  $\mathrm{MD}^D_g$ is a convex risk measure  if  and only if \begin{equation}\label{eq:conv_ES}\mathrm{MD}^D_g(X)=\lambda \E[(D(X)-Y)_+]+\E[X]\end{equation}   for some  non-negative random variable $Y\in L^1$ and some constant $\lambda\in [0,1]$.  In particular, $\mathrm{MD}^D_g$ is a  coherent risk measure if and only if $Y=0$. \end{proposition}
\begin{proof} 
We need to show that   $g$ is an increasing, convex function which satisfies $1$-Lipschitz condition  if and only if $g(x)=\lambda\E[(x-Y)_+]$ for some $Y\geq0$ and $0\le \lambda \leq1$. 
This is known in the literature; see 
  Theorems 1 and 6 of \cite{W56}.
\end{proof}

By Theorem \ref{thm:monetay}, we know that $\mathrm{MD}^D_g$ is a consistent risk measure, yet it fails to exhibit convexity when $g\in\mathcal G$ is not convex as shown in Theorem \ref{prop:2}. For instance, take $g(x)=\lambda \E[x\wedge Y]$ for some non-negative $Y$ and $\lambda\in[0,1]$. It is obvious that $g$  is concave and satisfies $1$-Lipschitz condition. In this case, $\mathrm{MD}^D_g(X)$ can be expressed as $$\mathrm{MD}^D_g(X)=\lambda\E[D(X)\wedge Y]+\E[X],$$ 
which is a consistent but not convex risk measure. Furthermore, Theorem \ref{prop:2} illustrates that $\mathrm{MD}^D_g$ is a convex but not coherent risk measure if $g\in\mathcal G$ is convex yet non-linear. This insight opens up a new perspective for constructing risk measures within the class of monotone mean-deviation risk measures. Specifically, it guides us in developing risk measures that are consistent yet not convex, or alternatively, convex but not coherent, all while possessing an explicit formulation. 
By assuming that $g(x)=\E[(x- Y)_+]$ or  $g(x)=\E[x\wedge Y]$  for some  non-negative random variable $Y$, we can construct  many convex  or consistent risk measures with explicit form which appear to be new in the literature.

\begin{example}\label{exm:1} Suppose that $g(x)=\E[(x- Y)_+]$  for some  $Y\geq0$ and $D\in  \overline{\mD}^p$ with some $p\in[1,\infty]$.

\begin{enumerate}[(i)] \item Let  $Y$  be the exponential distribution with  parameter $\beta>0$, that is, $\mathbb{P}(Y>y)=e^{-\beta y}$, then $g(x)=x+\left(e^{-\beta x}-1\right)/{\beta}$.  According to Proposition \ref{thm:2},  we have  $$\rho(X)=\E[X]+D(X)+\frac{1}{\beta}\left(e^{-\beta D(X)}-1\right),$$ which is a convex risk measure.
  
\item  Let  $Y$ follow a Pareto distribution with tail parameter $\theta>0$, that is, $\mathbb{P}(Y>y)=(1+y)^{-\theta}$ for $y \geq 0$. We have 
 $$g(x)=\left\{\begin{aligned}&x+\left((1+x)^{1-\theta}-1\right)/{(\theta-1)}, &\theta\neq1,\\ &x-\log(1+x),&\theta=1.
  \end{aligned} \right.$$
This yields
$$\rho(X)=\left\{\begin{aligned}&\E[X]+D(X)+\left(\left(1+D(X)\right)^{1-\theta}-1\right)/(\theta-1),&\theta\neq1,\\ &\E[X]+D(X)-\log(1+D(X)),&\theta=1,
  \end{aligned} \right.$$  and $\rho$  is a convex risk measure.
  \end{enumerate}
  \end{example}
  \begin{example}\label{exm:2} Suppose that   $g(x)=\E[x\wedge Y]$ for some $Y\geq0$ and $D\in  \overline{\mD}^p$ for some $p\in[1,\infty]$.
  
\begin{enumerate}[(i)]
    \item  Let $Y$ follow the exponential distribution with  parameter $\beta>0$,    we have $g(x)=\left(1-e^{-\beta x}\right)/\beta$.  Then it follows that 
  $$\rho(X)=\E[X]+\frac{1}{\beta}\left(1-e^{-\beta D(X)}\right),$$  which is a consistent risk measure but not a convex risk measure.
  
\item  Let  $Y$ follow a Pareto distribution with tail parameter $\theta>0$. We have  $$ g(x)=\left\{\begin{aligned}&\left(1-(1+x)^{1-\theta}\right)/{(\theta-1)}, &\theta\neq1,\\&\log (1+x), &\theta=1.
 \end{aligned} \right.$$ This yields
  $$\rho(X)=\left\{\begin{aligned}&\E[X]+\frac{ 1-(1+D(X))^{1-\theta}}{\theta-1},& \theta\neq1,\\&\E[X]+\log (1+D(X)),& \theta=1,\end{aligned}\right.$$  which  is a consistent risk measure but not a convex risk measure.
  \end{enumerate}
\end{example}

\section{Dual representation} \label{sec:4}
In this section, we  investigate the dual representation of monotonic deviation measures that are convex. Before showing the main result, we need some preliminaries. For $p\in[1,\infty)$, denote by $q$ the conjugate dual of $p$, i.e., $q=(1-1 / p)^{-1}$. 
Define $\mathcal A_p=\{Z\in L^q: Z\ge 0,~\E[Z]=1\}$.
For a convex $g\in \mathcal G$, we use $g^*$ to represent its conjugate function, i.e., $g^*(y)=\sup_{x\ge 0} \{xy-g(x)\}$. 
One can easily check that $g^*$ is increasing, convex and lower semicontinuous.  Note that $g:[0,\infty)\to\R$ is increasing and $1$-Lipschitz continuous with $g(0)=0$. Denote by $a=\lim_{x\to\infty} g'(x)\in[0,1]$ where $g'$ is the left derivative of $g$, and
we have 
$g^*(y)=0$ for $y\le 0$ and $g^*(y)=\infty$ for $y>a$. Hence, it holds that $g(x)=g^{**}(x)=\sup_{y\in[0,a]} \{xy-g^*(y)\}$ for $x\ge 0$ (see e.g., Proposition A.6 of \cite{FS16}).  For a range-normalized deviation measure $D$, denote by $R=D+\E$, which is a finite coherent risk measure on $L^p$. 
Moreover,  the following representation holds:
\begin{align}\label{eq-codual}
R(X)=D(X)+\E[X]=\max_{Z\in \mathcal A} \E[XZ],~~X\in L^p
\end{align}
for some   convex and weakly compact set $\mathcal A\subseteq \mathcal A_p$. 

\begin{theorem}\label{th-dual}
 Fix $p\in[1,\infty)$. 
Suppose that $g\in\mathcal G$ is convex with $\lim_{x\to\infty} g'(x)=a$ and  $D\in\mathcal {\overline D}^p$. We have 
\begin{align*}
\mathrm{MD}^D_g(X)=
\max_{Z\in\mathcal A}\left\{a\E[XZ]-g^*\left(\frac{a}{\sup\{\lambda\in[1,\infty): \lambda(Z-1)+1\in\mathcal A\}}\right)\right\}+(1-a)\E[X],~~X\in L^p,
\end{align*}
where $\mathcal A$ is defined in \eqref{eq-codual}.
\end{theorem}
\begin{proof}
By Theorem \ref{prop:2}, ${\rm MD}^D_g$ is a finite convex risk measure on $L^p$. It follows from Theorem 2.11 of \cite{KR09} that 
\begin{align*}
\mathrm{MD}^D_g(X)=\max_{Z\in\mathcal A_p}\{\E[XZ]-\beta(Z)\},~~X\in L^p,
\end{align*}
for some $\beta:L^q\to (-\infty,\infty]$ that is convex and lower semicontinuous, given by
\begin{align*}
\beta(Z)=\sup_{X\in L^p}\{\E[XZ]-\mathrm{MD}^D_g(X)\},~~Z\in L^q.
\end{align*}
We first aim to prove that 
\begin{align}\label{eq-penalty}
\beta(Z)=\begin{cases}
g^*\left(\frac{1}{\sup\{\lambda\in[1/a,\infty): \lambda(Z-1)+1\in\mathcal A\}}\right),~~& Z\in\mathcal Z,\\
\infty,~~& {\rm otherwise},
\end{cases}
\end{align}
where $\mathcal Z=\{aY+1-a: Y\in\mathcal A\}$.
For $Z\in L^q$, we have
\begin{align}\label{eq-dual1}
\beta(Z)
&=\sup_{X\in L^p}\{\E[XZ]-\mathrm{MD}^D_g(X)\} \notag 
\\& =\sup_{X\in L^p}\{\E[XZ]-\E[X]-g(D(X))\}\notag\\
&=\sup_{X\in L^p}\inf_{y\in[0,a]}\{\E[XZ]-\E[X]-D(X)y+g^*(y)\},
\end{align}
where we have used $g(x)=\sup_{y\in[0,a]} \{xy-g^*(y)\}$ in the last step. It holds that the objective function of \eqref{eq-dual1} is convex and lower semicontinuous in $y$ for any fixed $X$ since $g^*$ is convex and lower semicontinuous, and it is concave in $X$ for any fixed $y$. By a minimax theorem (see e.g., Theorem 2 of \cite{F53}), we have
\begin{align}\label{eq-dual2}
\beta(Z)
&=\inf_{y\in[0,a]}\sup_{X\in L^p}\{\E[XZ]-\E[X]-D(X)y+g^*(y)\}\notag\\
&=\inf_{y\in[0,a]}\sup_{X\in L^p}\{\E[XZ]-\E[X]-(R(X)-\E[X])y+g^*(y)\}\notag\\
&=\inf_{y\in[0,a]}\sup_{X\in L^p}\inf_{Y\in\mathcal A}\{\E[(Z-1+y-yY)X]+g^*(y)\},
\end{align}
where we have used \eqref{eq-codual} in the second and third steps. Obviously, the objective function of \eqref{eq-dual2} is convex and continuous with respect to the weak topology in $Y$ for any fixed $X$ and concave in $X$. 
Also note that $\mathcal A$ is convex and weakly compact. By the minimax theorem, we have
\begin{align}\label{eq-minimax2}
\beta(Z)=\inf_{y\in[0,a], Y\in\mathcal A}\sup_{X\in L^p}\{\E[(Z-1+y-yY)X]+g^*(y)\}.
\end{align}
Denote by $\widetilde{\mathcal Z}=\{yY+1-y: y\in[0,a],~Y\in\mathcal A\}$.
Note that the inner supremum problem above is infinite if $\p(Z-1+y-yY\neq 0)>0$ and is equal to $g^*(y)$ if $Z-1+y-yY=0$. We have that $\beta(Z)=\infty$ if $Z\in L^q\setminus \widetilde{\mathcal Z}$, and for $Z\in \widetilde{\mathcal Z}$,
\eqref{eq-minimax2} reduces to
\begin{align*}
\beta(Z)&=\inf\left\{g^*(y): y\in[0,a],~Y\in\mathcal A,~y(Y-1)+1=Z\right\}\\
&=\inf\left\{g^*(y): y\in[0,a],~\frac{Z-1}{y}+1\in\mathcal A\right\}\\
&=\inf\left\{g^*\left(\frac{1}{\lambda}\right): \lambda\in\left[\frac1a,\infty\right),~\lambda(Z-1)+1\in\mathcal A\right\}\\
&=g^*\left(\frac{1}{\sup\{\lambda\in[1/a,\infty): \lambda(Z-1)+1\in\mathcal A\}}\right),
\end{align*}
where the last step holds because $g^*$ is increasing.
To verify \eqref{eq-penalty}, it remains to show that $\mathcal Z=\widetilde{\mathcal Z}$, i.e., $\{aY+1-a: Y\in\mathcal A\}=\{yY+1-y: y\in[0,a],~Y\in\mathcal A\}$. It is clear that $\mathcal Z\subseteq \widetilde{\mathcal Z}$. Conversely, for any $Z\in \widetilde{\mathcal Z}$ with the representation $Z=yY+1-y$ for some $y\in[0,a]$ and $Y\in\mathcal A$, since $\mathcal A$ is convex and $1\in\mathcal A$, we have that $\mathcal Z$ is convex and $1\in\mathcal Z$. Note that $Z=(y/a)(aY+1-a)+(1-y/a)\cdot 1$, where $y/a\in[0,1]$ and  $aY+1-a\in \mathcal Z$. It holds that $Z\in\mathcal Z$. This yields the converse direction. Hence, we have verified \eqref{eq-penalty}.
Therefore, we have
$$
\mathrm{MD}^D_g(X)=\max _{Z \in {\mathcal Z}}\left\{\mathbb{E}[X Z]-g^*\left(\frac{1}{\sup \{\lambda \in[1 / a, \infty): \lambda(Z-1)+1 \in \mathcal{A}\}}\right)\right\}, \quad X \in L^p,
$$
where $\mathcal{Z}=\{a Y+1-a: Y \in \mathcal{A}\}$. Moreover, for $Z\in\mathcal Z$ with the form $Z=aY+1-a$, where $Y\in\mathcal A$, it holds that
 \begin{align*}&\E[X Z]-g^*\left(\frac{1}{\sup \{\lambda \in[1 / a, \infty): \lambda(Z-1)+1 \in \mathcal{A}\}}\right)\\&= \E[X( aY+1-a)]-g^*\left(\frac{1}{\sup \{\lambda \in[1 / a, \infty): \lambda a(Y-1)+1 \in \mathcal{A}\}}\right) \\&=a \E[XY]-g^*\left(\frac{a}{\sup \{\lambda \in[1, \infty): \lambda(Y-1)+1 \in \mathcal{A}\}}\right)+(1-a)\E[X].\end{align*} 
This completes the proof.
\end{proof}

Below we give two specific examples of Theorem \ref{th-dual} by choosing the coherent risk measure $R$ as ES or expectile (see e.g., \cite{NP87} and \cite{B14}), which are popular in practice. This choice results in two classes of $\mathrm{MD}_g^D$. 
\begin{example}
Let $R=\ES_\alpha$ with $\alpha\in(0,1)$,
$D=R-\E$,  $g\in\mathcal G$ be convex with $\lim_{x\to\infty} g'(x)=a$, and ${\rm MD}^D_g=g \circ D+\E$.
The well-known dual representation of ES in \citet[Example 4.40]{FS16} gives $R(X)=\max_{Z\in\mathcal A}\E[XZ]$ for $X\in L^1$ where $\mathcal A=\{Z\in\mathcal A_\infty: Z\le 1/(1-\alpha)\}$. Then
\begin{align*}
\sup\{\lambda\in[1,\infty): \lambda(Z-1)+1\in\mathcal A\}
&=\sup\left\{\lambda\in[1,\infty): \lambda(\esssup Z-1)+1\le \frac{1}{1-\alpha}\right\}\\
&=\frac{\alpha}{1-\alpha}(\esssup Z-1)^{-1}.
\end{align*}
By Theorem \ref{th-dual}, we obtain 
\begin{align*}
{\rm MD}^D_g(X)&=\max_{Z\in\mathcal A}\left\{a\E[XZ]-g^*\left(\frac{(1-\alpha)a}{\alpha}(\esssup Z-1)\right)\right\}+(1-a)\E[X]\\
&=\sup_{\gamma\in\left[1,\frac{1}{1-\alpha}\right]}\sup\left\{a\E[XZ]-g^*\left(\frac{(1-\alpha)(\gamma-1)a}{\alpha}\right):Z\in\mathcal A_\infty,~\esssup Z=\gamma\right\}+(1-a)\E[X]\\
&=\sup_{\gamma\in\left[1,\frac{1}{1-\alpha}\right]}\left\{a\ES_{1-\frac{1}{\gamma}}(X)-g^*\left(\frac{(1-\alpha)(\gamma-1)a}{\alpha}\right)\right\}+(1-a)\E[X]\\
&=\sup_{\gamma\in[0,\alpha]}\left\{a\ES_{\gamma}(X)-g^*\left(\frac{1-\alpha}{\alpha}\frac{\gamma a}{1-\gamma}\right)\right\}+(1-a)\E[X].
\end{align*}
Suppose now $a=1$, and we define $f:[0,1]\to(-\infty,\infty]$ as
\begin{align*}
f(\gamma)=\begin{cases}
g^*\left(\frac{1-\alpha}{\alpha}\frac{\gamma}{1-\gamma}\right),~~&\gamma\in[0,\alpha],\\
\infty,~~&\gamma\in(\alpha,1].
\end{cases}
\end{align*}
Obviously, $f$ is an increaing and convex function on $[0,1]$ as $g^*$ and $\gamma\mapsto \gamma/(1-\gamma)$ are both increasing and convex.
It holds that 
$$
{\rm MD}^D_g(X)=\sup_{\gamma\in[0,1]}\{\ES_\gamma(X)-f(\gamma)\}.
$$
{A functional of the form  $ \sup_{\gamma\in[0,1]}\{\ES_\gamma(X)-h(\gamma)\} $  for a general function $h$ is  called an \emph{adjusted Expected Shortfall} (AES)  by \cite{BMW22}.
Different from the general class of AES considered by \cite{BMW22}, 
the subclass  $\mathrm{MD}_g^D$ has an explicit formula, i.e., ${\rm MD}^D_g(X)=g(\ES_{\alpha}(X))+\E[X]$.}
\end{example}

\begin{example}
An expectile at level $\alpha \in (0,1)$, denoted by ${\rm ex}_\alpha$, is defined as the solution of the following equation:
\begin{align*}
\alpha \E[(X-x)_+]= (1-\alpha) \E[(X-x)_-],~~X\in L^1.
\end{align*}
When $\alpha\ge 1/2$, ${\rm ex}_\alpha$ is a convex risk measure admitted a dual  representation (see e.g., Proposition 8 of \cite{B14}):
\begin{align*}
{\rm ex}_\alpha(X)=\max_{Z\in\mathcal A}\E[XZ]~~~{\rm with}~\mathcal A=\left\{Z\in \mathcal A_\infty: \frac{\esssup Z}{\essinf Z}\le \frac{\alpha}{1-\alpha}\right\}.
\end{align*}
Let $R={\rm ex}_\alpha$ with $\alpha\in[1/2,1)$,
$D=R-\E$ and $g\in\mathcal G$ be convex with $\lim_{x\to\infty} g'(x)=a$, and let ${\rm MD}^D_g=g \circ D+\E$. It holds that
\begin{align*}
\sup\{\lambda\in[1,\infty): \lambda(Z-1)+1\in\mathcal A\}
&=\sup\left\{\lambda\in[1,\infty): \frac{\lambda (\esssup Z-1)+1}{\lambda(\essinf Z-1)+1}\le \frac{\alpha}{1-\alpha}\right\}\\
&=\frac{2\alpha-1}{2\alpha-1+(1-\alpha)\esssup Z-\alpha\essinf Z}.
\end{align*}
By Theorem \ref{th-dual}, we obtain
\begin{align*}
{\rm MD}^D_g(X)=\sup_{Z\in\mathcal A}\left\{a\E[XZ]-g^*\left(\frac{a((1-\alpha)\esssup Z-\alpha\essinf Z)}{2\alpha-1}+a\right)\right\}+(1-a)\E[X].
\end{align*}
\end{example}

Recall that $g\in\mathcal G$ is convex with $\lim_{x\to\infty} g'(x)=a$ and $D\in\mathcal {\overline D}^p$ for some $p\in[1,\infty)$ with $R=D+\E$ defined in \eqref{eq-codual}. 
Theorem \ref{th-dual} illustrates that the smallest coherent risk measure that dominates the convex risk measure ${\rm MD}_g^D$ is $a D +\E$. Below we give an analogous result for ${\rm MD}^D_g$ where $g\in\mathcal G$ is not necessarily convex.

\begin{proposition}\label{prop:smCoherant}
Let $g\in\mathcal G$ and $D\in \overline{\mathcal D}^p$. The smallest coherent risk measure that dominates ${\rm MD}^D_g$ is $(\sup_{x>0}g(x)/x) D(X)+\E[X]$.
\end{proposition}

\begin{proof}
The smallest positive homogeneous functional that dominates ${\rm MD}^D_g$ is given by
\begin{align*}
\rho(X)=\sup_{\lambda>0} \frac{{\rm MD}^D_g(\lambda X)}{\lambda}
&=\sup_{\lambda>0} \frac{g(\lambda D(X))}{\lambda}+\E[X]\\
&=D(X) \sup_{\lambda>0}\frac{g(\lambda)}{\lambda}+\E[X]. 
\end{align*}
Hence, we get the desired result.
\end{proof}

Proposition \ref{prop:smCoherant} addresses the case where $g$ is convex, which aligns with the observation in Theorem \ref{th-dual}. This is because $\sup_{x>0}g(x)/x=\lim_{x\to\infty} g'(x)$ for convex $g\in\mathcal G$.
Despite the simplicity of the proof of above proposition,
the smallest dominating coherent risk measure of a given risk measure has several interesting applications; see \cite{WBT15} in the context of subadditivity,
and \cite{HK22b} in the context of arbitrage induced by risk measure.

\section{Non-parametric estimation}\label{sec:6}
In this section, we consider the properties of non-parametric estimators of $\mathrm{MD}^D_g$  when  $D$ is   a  convex signed Choquet integral. 
Denote by  $$\mathcal{H}=\{h: h ~\text{is~a concave~function from}~ [0,1] ~\text{to}~ \mathbb{R} ~\text{with}~ h(0)=h(1)=0\}.$$
Since $h\in\mathcal H$ is concave, its left derivative $h'$ is well defined almost everywhere. If $h$ is further continuous on $[0,1]$, we denote by $\left\|h^{\prime}\right\|_q$ the $q$-Lebesgue norm of $h^{\prime}$, i.e., $\left\|h^{\prime}\right\|_q=(\int_0^1|h^{\prime}(t)|^q \mathrm{~d} t)^{1 / q}$ for $q\in[1,\infty)$ and $\|h'\|_\infty=\sup_{t\in(0,1)}|h'(t)|$. If $h$ is not continuous, we adopt the convention that $\|h'\|_q=\infty$ for all $q\in[1,\infty]$.
By Theorem 2.4 of \cite{LCLW20},  for $D\in  \mD^p$ and $p \in[1, \infty)$, there exists a set $\Psi^p \subseteq \Phi^p$ such that    \begin{equation}\label{eq:h}
 D(X)=\sup _{h \in \Psi^p}\left\{\int_0^1 \VaR_\alpha(X) h'(1-\alpha)\mathrm{d}\alpha \right\}, \quad X \in L^p,
 \end{equation}
where
 $ 
 \Phi^p=\{h \in \mathcal H:\left\|h^{\prime}\right\|_q<\infty\}$ with $q=(1-1 / p)^{-1}$.   For $h \in  \mathcal{H}$, the mapping   \begin{equation}\label{eq:D_h} D_h(X)=\int_0^1 \mathrm{VaR}_\alpha(X)  h'(1-\alpha)\mathrm{d} \alpha=\int_\R h 
 \left( \p(X> x)\right)\d x,
 \mbox{~~~~
  $X \in L^p$,}
 \end{equation} is  a   signed Choquet integral, characterized {by  \cite{WWW20a,WWW20b} via   comonotonic additivity.}\footnote{ Random variables $X$ and $Y$ are said to be comonotonic if there exists $\Omega_0 \in \mathcal{A}$ with $\mathbb{P}\left(\Omega_0\right)=1$ such that $\omega, \omega^{\prime} \in \Omega_0$
$
\left(X(\omega)-X\left(\omega^{\prime}\right)\right)\left(Y(\omega)-Y\left(\omega^{\prime}\right)\right) \geq 0.
$
 For a functional $\rho: \mathcal{X} \rightarrow \mathbb{R}$, we say that $\rho$ is comonotonic additive, if for any comonotonic random variables $X, Y \in \mathcal{X},~\rho(X+Y)=\rho(X)+\rho(Y)$.} The function $h$ is called the distortion function of $D_h$.  
{\color{black}  By Theorem 3 of \cite{WWW20b}, for any function $ h $ mapping $[0,1]$ to $\mathbb{R}$ with bounded variation and $ h(0)=0 $, $ h $ is concave if and only if $ D_h $ is subadditive. Combining this with the facts that $ h \in \mathcal H$ is always non-negative and $ \int_0^1 c h'(1-\alpha) \, \mathrm{d}\alpha = 0 $ for any $c\in\R$,  we can check  that $D$ in (16) and $ D_h $ in (17) satisfy all four properties of deviation measures in Definition 1.}
 
 By Proposition 1 of \cite{WWW20a},   $D_h$ is finite on $L^p$ for $p \in[1, \infty)$  if $h\in \Phi^p$, and $D_h$ is always finite on $L^{\infty}$.  In particular,  for $D\in  \mD^p$, if  $D$ is comonotonic additive,   then $D$  can only be the signed Choquet integrals; see  Theorem 1 of \cite{WWW20a}.
{
For $h\in\Phi^p$ with $p\in[1,\infty)$ and $D_h: L^p\to\R$ defined by \eqref{eq:D_h}, one can observe that $\lambda D_h+\E$ satisfies [CA], [PH], [SA] for any $\lambda\ge 0$. Combining Proposition 2 (ii) of \cite{WWW20a} and Proposition \ref{prop:verify}, it is established that $D_h$ is range normalized if and only if $t\mapsto h(t)+t$ is increasing on $[0,1]$ and $t\mapsto\lambda h(t)+t$ is not an increasing function on $[0,1]$ for any $\lambda>1$, and this is further equivalent to $h'(1)=-1$; we do not assume this condition in this section.
}


The  non-parametric estimators of $\mathrm{MD}^D_g(X)$  can be derived from those of $D_h$, VaR and the expectation, as we will explain in this section.
For $p\in[1,\infty]$, suppose that $X_1, X_2, \ldots, X_n \in  L^p$ are  an iid sample from (the  
 distribution of) a  random variable $X$. Recall that the empirical distribution $\widehat{F}_n$ of $X_1, \ldots, X_n$ is given by
$$\widehat{F}_n(x)=\frac{1}{n} \sum_{j=1}^n \id_{\left\{X_j \leq x\right\}}, \quad x \in \mathbb{R} .$$
Let $\widehat{\mathrm {MD}}^D_g(n)$ be the empirical estimator of $\mathrm {MD}^D_g(X)$, obtained by applying $\mathrm {MD}^D_g$ to the empirical distribution of $X_1, \ldots, X_n$.  We will establish consistency and asymptotic normality of the empirical estimators, based on corresponding results on empirical estimators of $\E[X]$ and $D_h(X)$. Let $\widehat x_n$ and $\widehat D_h(n)$ be the empirical estimators of $\E[X]$ and $D_h(X)$ based on the first $n$ sample data points.  We make  following standard regularity assumption on the distribution of the random variable $X$.
\begin{assumption}\label{assump:1}
The distribution $F$ of $X \in \X$ is supported on a convex set and has a positive density function $f$ on the support. Denote by $\tilde f=f \circ F^{-1}$. 
\end{assumption}
The proof of Theorem \ref{thm:4} below relies on standard techniques in empirical quantile processes, and it is given in Appendix \ref{app:pf_thm5}. 
{In what follows, $g'$ is the left derivative of $g$.} 
\begin{theorem}\label{thm:4} Fix $p\in[1,\infty)$. Let $g\in \mathcal G$ and $D=D_h$  where $h\in \Phi^p$.  Suppose that $X_1,\dots, X_n \in L^p$ {are} an iid sample from  $X\in L^p$ and Assumption \ref{assump:1} holds. Then,  $g(\widehat D(n))+\widehat x_n\stackrel{\mathbb{P}}{\rightarrow}g\left(D(X)\right)+\E[X]$ as $n\to\infty$. Moreover, if $p<2$ and $X  \in L^{\gamma}$ for some $\gamma>2p/(2-p)$, then 
we have   $$
\begin{aligned}
\sqrt{n}\left(\widehat{\mathrm{MD}}^D_g(n)-\mathrm{MD}^D_g(X)\right) \stackrel{\mathrm{d}}{\rightarrow} \mathrm{N}\left(0, \sigma_g^2\right),
\end{aligned}
$$ in which \begin{equation}\label{eq:sigma_g}\begin{aligned}\sigma^2_g= \int_0^1 \int_0^1 \frac{(h'(1- s)g'(D(X)) +1)(h'(1- t)g'(D(X))+1)(s \wedge t-s t)}{\tilde f(s) \tilde f(t)} \mathrm{d} t \mathrm{d} s. \end{aligned}\end{equation}
\end{theorem}

The integrability conditions $h\in \Phi^p$ and $X\in L^{\gamma}$ with $\gamma>2p/(2-p)$, needed for asymptotic normality in Theorem \ref{thm:4}, coincide with those in \cite{JZ03}, who gave asymptotic normality of empirical estimators for distortion risk measures. In particular, in case $p=1$, we require $X\in L^\gamma$ with $\gamma>2$, which is a common assumption in weighted empirical quantile processes without distortion; see e.g., \cite{SY96}. 
The condition $p<2$ is also important. If $D_h\not \in \Phi^2$, then $D_h$ is not even finite on $L^2$, and we do not expect asymptotic normality in this case.

Note that the asymptotic variance $\sigma_g$ in \eqref{eq:sigma_g}
is decreasing in the left derivative $g'$  of $g$. Therefore, if we replace $g$ by $\tilde g\in \mathcal G$ satisfying $\tilde g'\le g'$, then the asymptotic variance, and thus the estimation error, will decrease. 
Note that a larger  $g'$ corresponds to a larger sensitivity to risk, as it measures how $\mathrm{MD}^D_g$ changes when $D(X)$ increases.
Therefore, Theorem \ref{thm:4} gives a trade-off between risk sensitivity and statistical efficiency.

In what follows, we present some simulation results  based on Theorem \ref{thm:4}.  We   assume that 
$g(x)=x+e^{-x}-1$ and $g(x)=1-e^{-x}$, respectively.  Simulation results are presented  in the case of standard normal
and Pareto risks with tail index 4.  Let  the sample size $n=10000$, and  we repeat   the procedure  $5000$ times.   

First, let $D=\ES_\alpha-\E$ with  $\alpha=0.9$,  then $\mathrm{MD}^D_g(X)$  is given as   $\mathrm{MD}^D_g(X)=g(\ES_\alpha(X)-\E[X])+\E[X].$    In this case, we have $h'(1-t)=\frac{1}{1-\alpha}\id_{\{t\geq \alpha\}}-1$ and $\sigma^2_g$ in \eqref{eq:sigma_g} can be computed explicitly. 
We compare the asymptotic variance   of $\mathrm{MD}^D_g$ with that of $\ES_\alpha$, given by,  via \eqref{eq:sigma_g},
$$\begin{aligned} \sigma^2_{\mathrm{ES}}=\frac{1}{(1-\alpha)^2} \int_\alpha^1 \int_\alpha^1 \frac{s \wedge t-s t}{\tilde f(s) \tilde f(t)} \mathrm{d} t \mathrm{d} s \end{aligned}.$$

In Figure \ref{fig:normal1} (a) and (b), the sample is simulated from standard normal risk. We can observe that, for  $g(x)=x+e^{-x}-1$ and  $D=\ES_\alpha-\E$,  empirical estimates of  $\mathrm{MD}^D_g$ match quiet well with the density function of $\mathrm{N}(0.93,2.85/n)$.  In contrast, $\ES_\alpha$  empirical estimates match  with the density function of  $\mathrm{N}(1.76,3.71/n)$, whose  asymptotic variance    is larger than that of  $\mathrm{MD}^D_g$. In Figure \ref{fig:normal1} (c) and (d), the sample is simulated from the Pareto distribution with tail index $4$. We can observe that  $\mathrm{MD}^D_g(X)$ empirical estimates match quiet well with the density function of $\mathrm{N}(0.73,4.88/n)$ and  ES  empirical estimates match  with the density function of  $\mathrm{N}(1.37,10.19/n)$, whose     asymptotic variance    is  also larger than the one of  $\mathrm{MD}^D_g$. Since $g$ satisfies the 1-Lipschitz condition, the volatility of $D$ is reduced via the distortion by $g$.

\begin{figure}[htb!]
\centering
 \includegraphics[width=16cm]{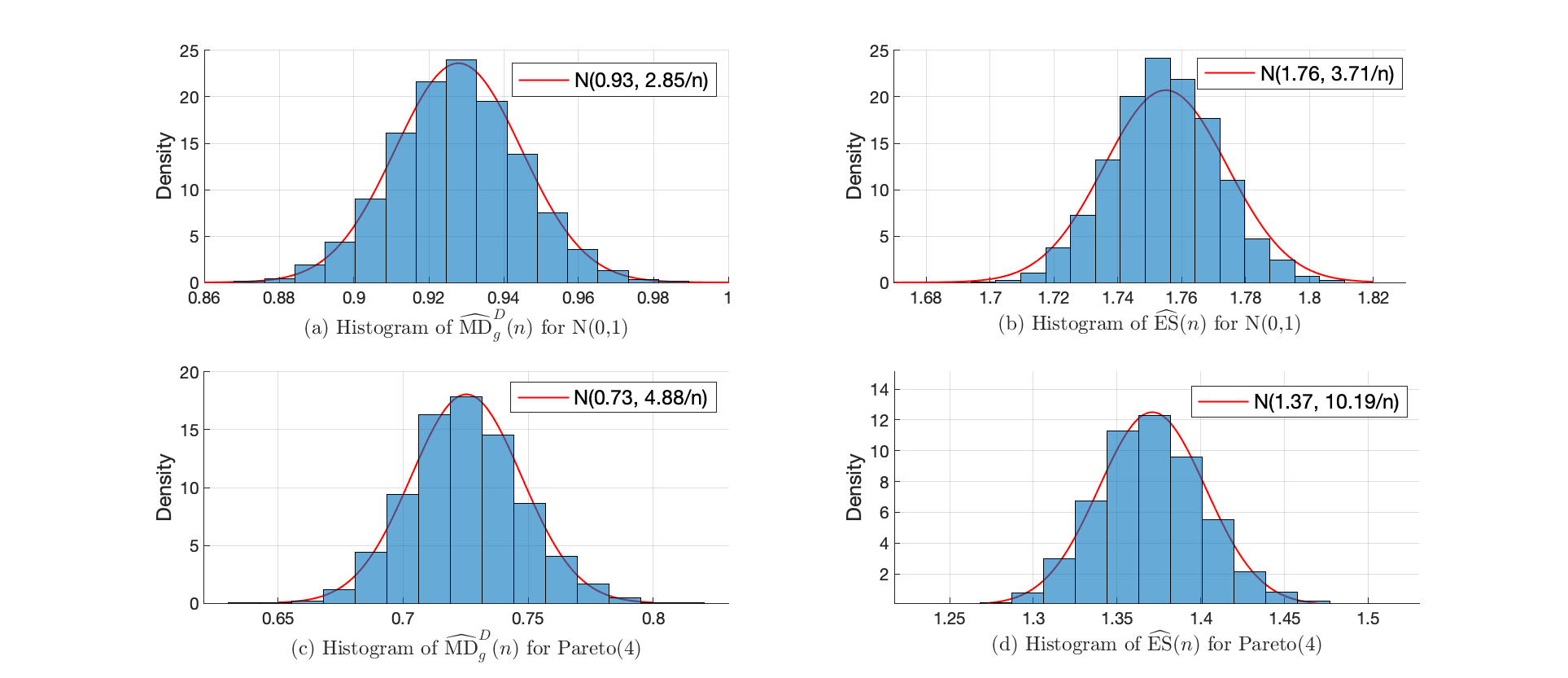}
 \captionsetup{font=small}
{\caption{ Left: $\widehat{\mathrm{MD}}^D_{g}(n)$  with $D=\ES_\alpha-\E$ and $g(x)=x+e^{-x}-1$; Right: $\widehat\ES_\alpha(n)$}\label{fig:normal1}}
\end{figure}

In Figure \ref{fig:normal2} (a) and (b),  for  $g(x)=1-e^{-x}$ and  $D=\ES_\alpha-\E$,  we can observe that  ${\mathrm{MD}}^D_g$ empirical estimates match quiet well with the density function of $\mathrm{N}(0.83,1.08/n)$ and $\mathrm{N}(0.98,1.97/n)$ when the samples are also  simulated from the standard normal distribution or  the Pareto distribution with tail index $4$.   Also, the asymptotic variance    are both  smaller than  those of  ES  empirical estimates.
\begin{figure}[htb!]
\centering
 \includegraphics[width=16cm]{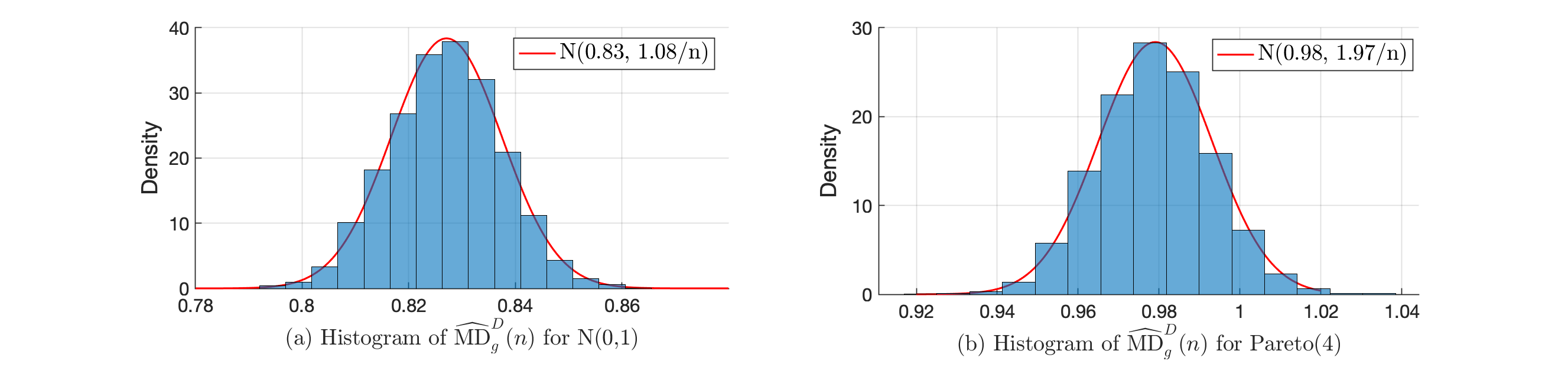} \captionsetup{font=small}
 \caption{ \small $\widehat{\mathrm{MD}}^D_{g} (n)$  with $D=\ES_\alpha-\E$ and $g(x)=1-e^{-x}$}\label{fig:normal2}
\end{figure}

If   $g(x)=\lambda x$ with $\lambda\in(0,1)$,   then we have $\mathrm{MD}^D_g(X)=\lambda \ES_\alpha(X)+(1-\lambda)\E[X]$.  It is obvious that the  asymptotic variance of $\E/\ES$-mixture is an increasing function with respect to $\lambda$ and thus is smaller than the one of $\ES$.  Moreover, if $\lambda=1$,  $\mathrm{MD}^D_g=\ES_\alpha$, and the values of $\sigma_g^2/n$ in Figure \ref{fig:normal3} (b) and (d) equal to those in Figure \ref{fig:normal1} (b) and (d).

\begin{figure}[htb!]
\centering
 \includegraphics[width=16cm]{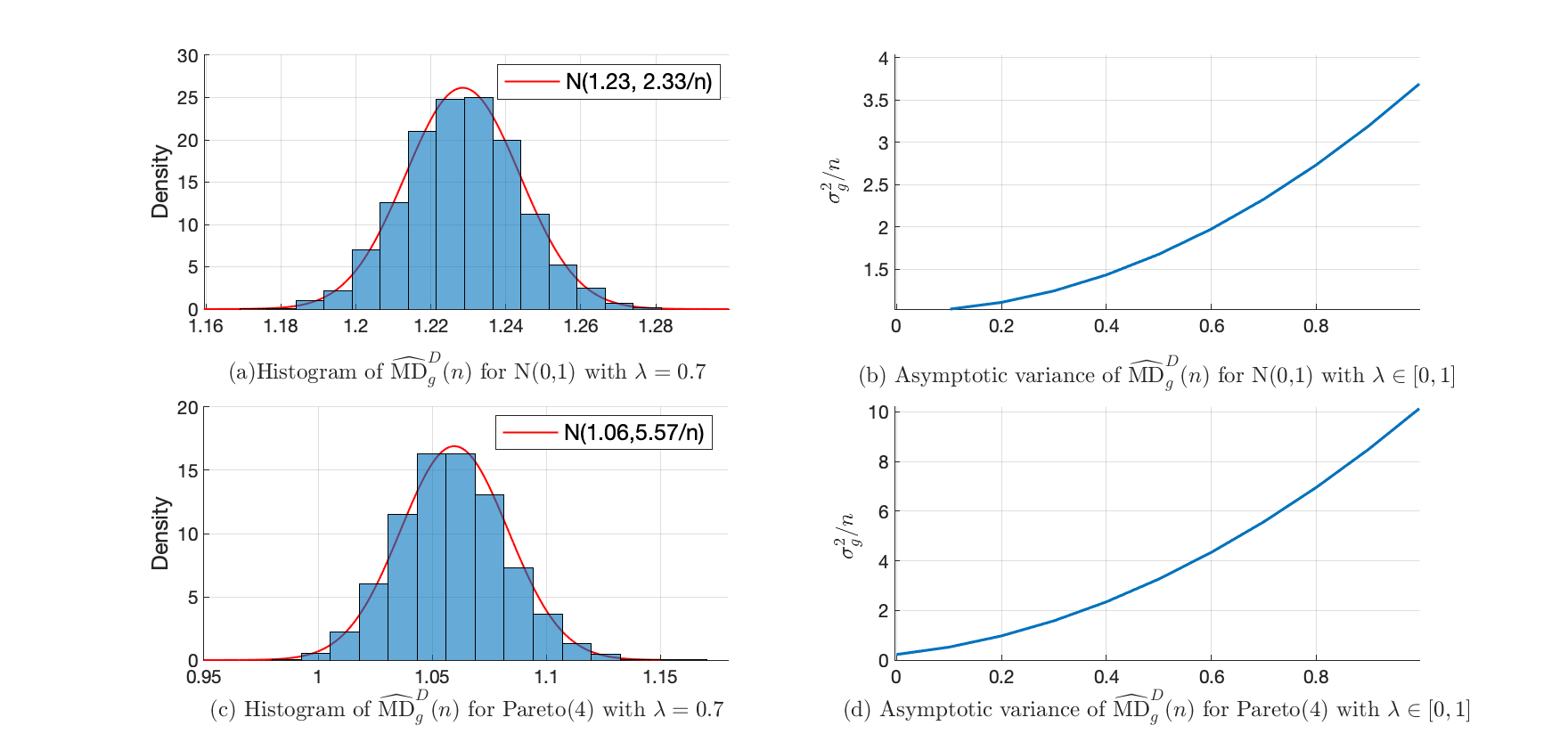}
  \captionsetup{font=small}
 \caption{ \small $\widehat{\mathrm{MD}}^D_{g}(n)$ with   $D=\ES_\alpha-\E$ and $g(x)=\lambda x$}\label{fig:normal3}
\end{figure}

Below we give another example of $D$. For $X \in L^{1}$, let  $X_1$, $X_2$, $X$  be iid, and
\begin{equation}
    D(X)=\mathrm{Gini}(X):=\frac{1}{2} \mathbb{E}\left[\left|X_1-X_2\right|\right].
\label{eq:Gini-d}
\end{equation}
The Gini deviation is a signed Choquet integral with a concave distortion function $h$ given by $h(t)=t-t^2$ for $t \in[0,1]$ (see e.g., \cite{D90}), i.e., $D={\rm Gini}=D_h$.
Then we have $$\sigma^2_g= \int_0^1 \int_0^1 \frac{((2s-1)g'(\mathrm{Gini}(X)) +1)((2t-1)g'(\mathrm{Gini}(X))+1)(s \wedge t-s t)}{\tilde f(s) \tilde f(t)} \mathrm{d} t \mathrm{d} s. $$
  Note that  the asymptotic variance  for  $\mathrm{Gini}(X)+\E[X]$, denoted by $\sigma^2_{\mathrm{Gini}+\E}$, equals $$\begin{aligned} \sigma^2_{\mathrm{Gini}+\E}=\int_0^1 \int_0^1 \frac{4ts(s \wedge t-s t)}{\tilde f(s) \tilde f(t)} \mathrm{d} t \mathrm{d} s \end{aligned}.$$ Simulation results are presented in Figures \ref{fig:Dini1} and \ref{fig:Dini2}  for $D=\rm{Gini}$ in the case of the standard normal distribution
and the Pareto distribution with tail index 4, that also confirm the asymptotic normality of the empirical
estimators in Theorem \ref{thm:4}.  Similarly,  the     asymptotic variance  of $\E+\mathrm{Gini} $  is  also larger than the one of  ${\mathrm{MD}}_{g}^D$ based on $D=\rm Gini$.  
\begin{figure}[htb!]
\centering
 \includegraphics[width=16cm]{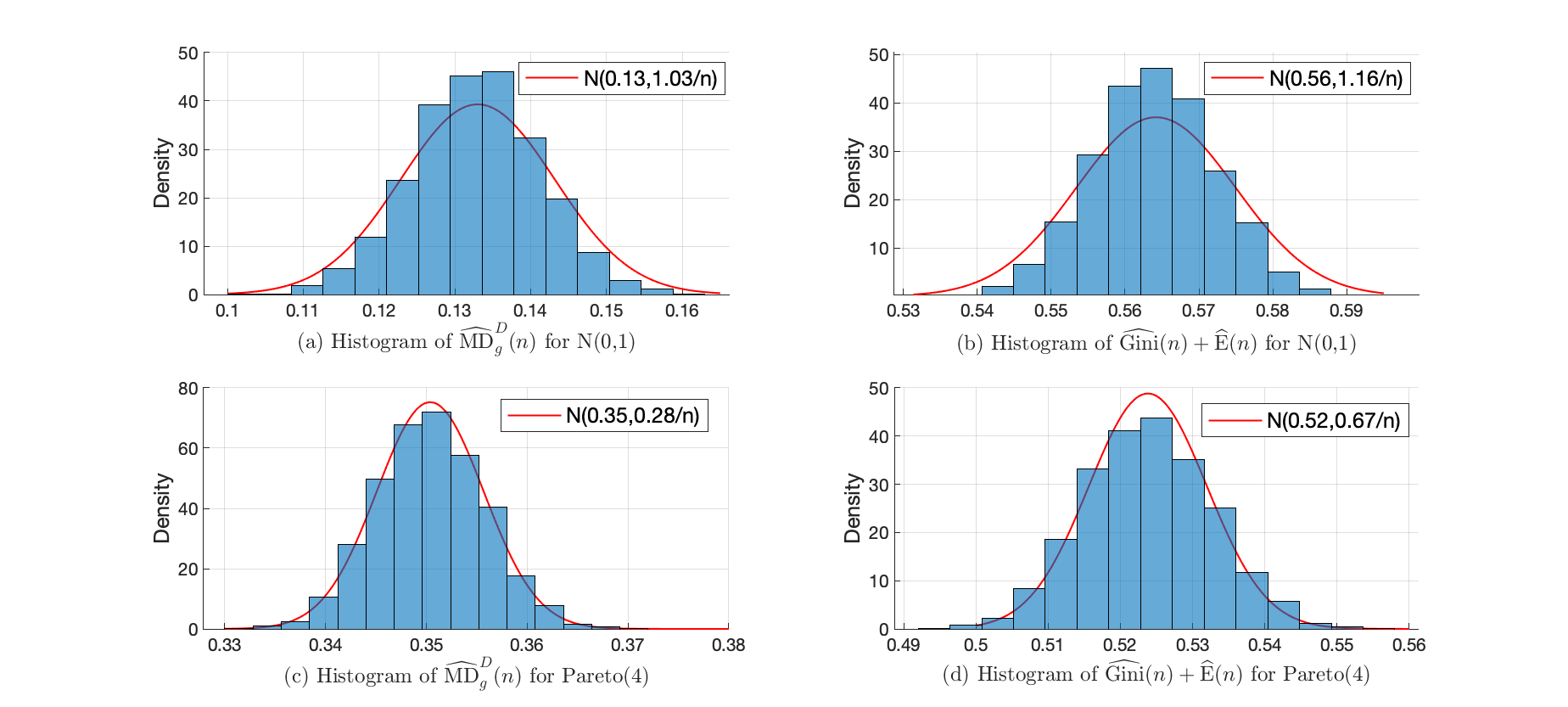}
 \captionsetup{font=small}
 \caption{ \small Left:  $\widehat{\mathrm{MD}}^D_{g}(n)$  with  $g(x)=x+e^{-x}-1$ and $D=\mathrm{Gini}$;
 Right: $\widehat{\mathrm{Gini}}(n)+\widehat{\E}(n)$}\label{fig:Dini1}
\end{figure}
\begin{figure}[htb!]
\centering
 \includegraphics[width=16cm]{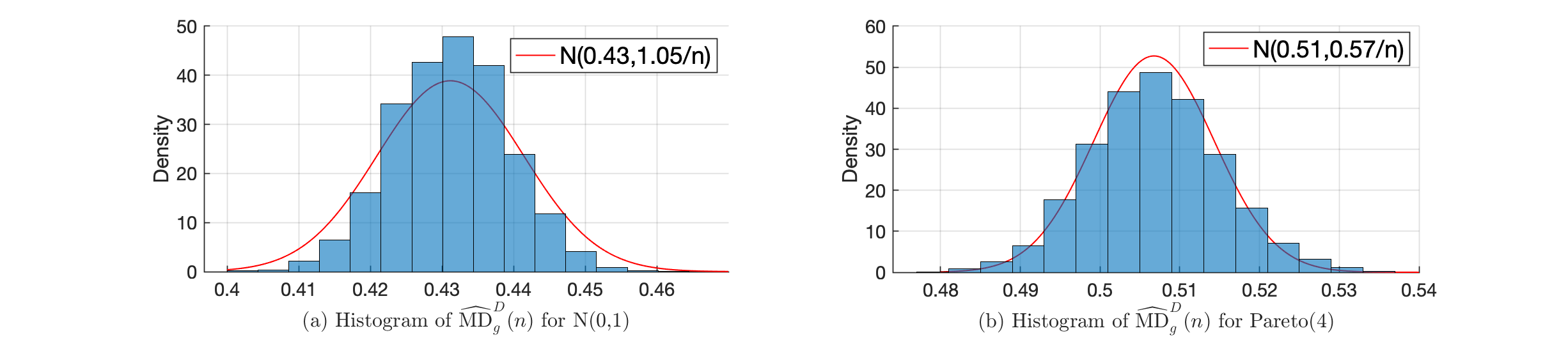}
 \captionsetup{font=small}
 \caption{ \small   $\widehat{\mathrm{MD}}^D_{g}(n)$  with  $D=\mathrm{Gini}$ and $g(x)=1-e^{-x}$}\label{fig:Dini2}
\end{figure}

{\color{black}
\section{Application to portfolio selection}\label{sec:opt}
 
In this section,  we consider portfolio selection problems based on $\mathrm{MD}_g^D$.
Let a 
random vector  $\mathbf{X}\in\X^n$ represent log-losses  (i.e., the negation of log-return of the daily asset prices; see \cite{MFE15}) from  $n$ assets and a vector $\mathbf w= (w_{1},  \dots, w_{n}) \in \Delta_{n}$ of portfolio weights, where 
$\Delta_{n}$ is the standard $n$-simplex,  given by
$$  \Delta_{n}=\left\{(w_{1},  \dots, w_{n}) \in[0,1]^{n}: w_{1} +\dots+w_{n}=1\right\}.$$ 
The total loss  of the portfolio is  $\mathbf{w}^\top \mathbf{X}$,  and the optimization problem  is formulated as \begin{equation}\label{eq:optimal}
\min _{\mathbf{w} \in \Delta_n} \mathrm{MD}_{g}^D(\mathbf{w}^{\top} \mathbf{X}).
\end{equation}
Note that convexity of $g$ implies that $\mathrm{MD}_{g}^D$ is a convex risk measure (see Theorem \ref{prop:2}), and in this case, problem \eqref{eq:optimal} is a convex optimization problem.


    We select the 4 largest stocks from each of the 10 different sectors of S\&P 500,    ranked by market cap in the beginning of 2014,  as the portfolio compositions (40 stocks in total).    The historical   asset prices are collected from Yahoo Finance,    covering the period from January 2, 2014,  to December 29, 2023, with a total of 2516 observations of daily losses.   We use the first two years' data for training and start investment strategies at the beginning of 2016. The initial wealth is set to $1$, and the risk-free rate  is $r=2.13\%$, which is the 10-year yield of the US treasury bill in Jan 2016.  Note that  the risk-free asset is not used to construct portfolios, but only used to calculate the Sharpe ratios.

Our portfolio strategies rebalance at the beginning of each month by solving \eqref{eq:optimal} and we assume no transaction cost.
At each period, we use the empirical distribution of the previous 500 log-loss data to estimate the risk measure, i.e., using an empirical estimator as in Section \ref{sec:6}.
This is the simplest method of computing the risk measure, although the standard method in practice is to fit a time-series model. We choose this simple method for illustrative purposes. 

We consider $\mathrm{MD}_{g}^D$
 by choosing $D=\ES_\alpha-\E$ with $\alpha=0.9$ and varying   $g$, since our main novelty lies in the risk weighting 
 function $g$. In particular, we consider the following   class of convex functions $g_
 \beta$ indexed by a parameter $\beta>0$ as in  Example \ref{exm:1} (i), given by \begin{equation}\label{eq:g_exp}g_
 \beta(x)=x+\frac 1 {\beta} \left(e^{-\beta x}-1\right).\end{equation} 
 The parameter $\beta$ has a natural interpretation of describing the convexity of $g_
 \beta$; that is, a smaller $\beta$ means a more convex $g_\beta$. This is because $g_\beta''/g_\beta'$ is decreasing in $\beta$ (see \cite{R81} for comparing convexity of functions). Note that $g_\beta(x)\to x$ as $\beta\to \infty$, which represents a linear risk weighting function.

 At each period, the problem is to minimize $\mathrm{MD}_{g}^D$ over $\mathbf w\in \Delta_n$, that is, 
 $$\min_{\mathbf w\in \Delta_n}:~~ \E[\mathbf w^\top \mathbf X]+ 
 g_\beta (\ES_\alpha(\mathbf w^\top \mathbf X) -\E[\mathbf w^\top \mathbf X]),
 $$
 where $\mathbf X$ follows the empirical distribution of the log-loss vector of the previous 500 trading days. 
 By using the ES optimization formula of \cite{RU02}, that is,
 $$
 \ES_\alpha(X) = \min_{x\in \R}\left\{x+ \frac 1 {1-\alpha}\E[(X-x)_+]\right\}, ~~X\in L^1,
 $$
 we can write the $\mathrm{MD}_{g}^D$ minimization problem as 
\begin{align}
    \label{eq:optimize-port}
    \min_{\mathbf w\in \Delta_n,~ x\in \R}:~~  \E[\mathbf w^\top \mathbf X]+ 
 g_\beta\left (x+ \frac{1}{1-\alpha} \E[(\mathbf w^\top \mathbf X-x)_+]-\E[\mathbf w^\top \mathbf X]\right).
\end{align}
The problem \eqref{eq:optimize-port} is jointly convex in $\mathbf w$ and $x$ and therefore can be easily solved numerically by modern computational programs such as MATLAB.

We choose $\beta =1,3,10,30,100$ to study the effect of $\beta$.
We  compare them with a portfolio that simply minimizes $\ES_\alpha$ (corresponding to $\beta=\infty$) and a \cite{M52} portfolio (by fixing an expected log-return at 10\% and minimizing the variance at each period).
 The portfolio performance is reported in Figure \ref{sec_40}.    Summary statistics,  including the annualized return (AR), the annualized volatility (AV), and the Sharpe ratio (SR)  are reported  in Table \ref{tab_40}.

\begin{table}[h]  
\def\arraystretch{1}
  \begin{center} 
    \caption{Annualized return (AR),  annualized volatility (AV), and Sharpe ratio (SR)  for   different portfolio strategies from Jan 2016 to Dec 2023} \label{tab_40} 
  \begin{tabular}{c|ccccccc} 
    $\%$ & $\beta=1$ &  $\beta=3$     & $\beta=10$ & $\beta=30$    & $\beta=100$  &$\ES_\alpha$ & MV  \\    \hline
AR & 9.28    &   9.05  &9.23 &  9.14 & 8.71& 8.41 &6.77 \\
AV & 19.92    & 15.21 & 12.78 & 12.26&12.05&12.10&11.47 \\
SR  & 35.89     & 45.47& 55.54&  57.19 & 54.59 & 51.95 &38.76  \\
 \hline \hline 
    \end{tabular}
    \end{center}
    \end{table}

 \begin{figure}[h]
\caption{Wealth processes for  portfolios, 40 stocks, Jan 2016 - Dec 2023}\label{sec_40} 
\centering
           \includegraphics[width=14cm]{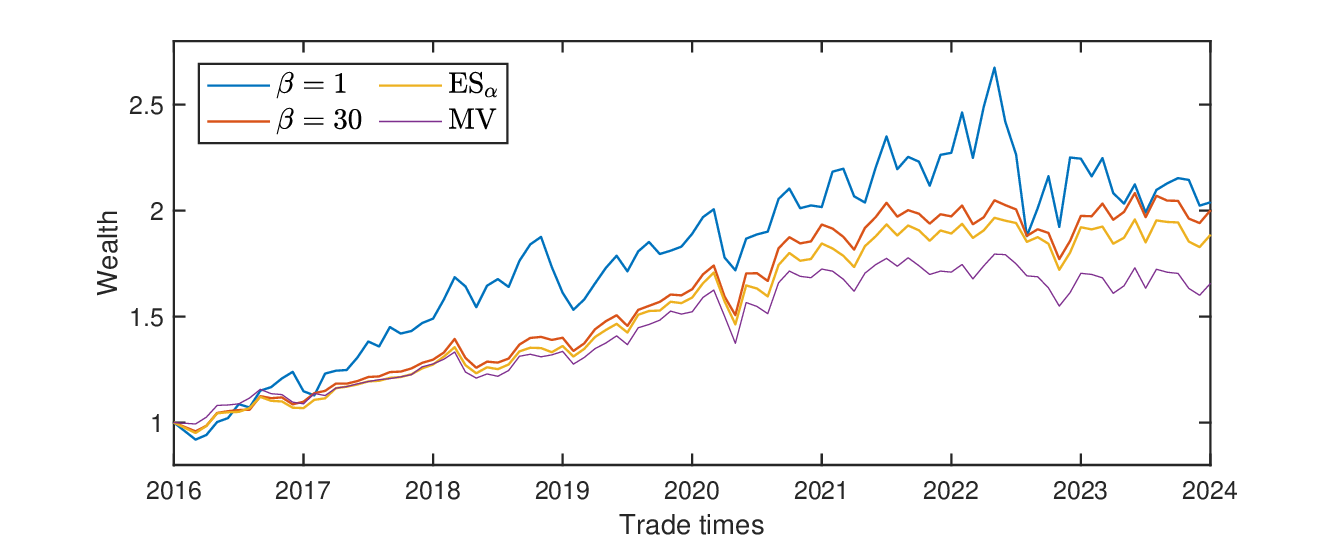}
\end{figure}

Our findings suggest that $\mathrm{MD}_g^D$  minimizing portfolios have some intuitive features. 
A smaller $\beta$, meaning a more convex and smaller risk weight function $g_\beta$, roughly leads to a larger annualized return and a larger annualized volatility. 
This is intuitive because the weight on the deviation measure $\ES_\alpha-\E$ is smaller for smaller $\beta$, thus putting higher value on the return. On the other hand, a large Sharpe ratio is attained around $\beta=30$, suggesting that a suitable level of  $\beta$ can balance the return and the volatility quite well. The more convex $g_\beta$ is, the more MD neglects small values of deviation. This may partially explain the observation in Figure \ref{sec_40}, where the   curve corresponding to $\beta=1$ has many more fluctuations than the  one corresponding to $\beta =30$. 
We admit that this observation does not have a theoretical justification, and it is based only on one dataset, so we do not intend to generalize. 
To fully understand the effect of convexity of $g$ on portfolio selection,  future studies are needed. 

Of course, our objective is not to identify which strategy yields the highest return or Sharpe ratio in financial practice, a question that depends highly on the market situation and economic environment.  Instead, the empirical results presented here mainly illustrate the interpretability of the strategy for portfolio selection based on ${\rm MD}_g^D$.  
Moreover, the optimization and portfolio strategies are easy to implement.

 }
    \section{Conclusion} \label{sec:7}
Even though mean-deviation measures are  widely considered in the literature and have a lot of  attractive features, there are few systemic treatments in the literature.   In this paper, we studied  the class $\mathrm{MD}^D_g$ of  mean-deviation measures whose  form is  a  combination of  the  deviation-related functional and the expectation,  which enriches the axiomatic theory of risk measures.      In particular, the obtained class  always belongs to the class of  consistent risk measures.  We showed that  $\mathrm{MD}^D_g$ can be
coherent, convex or star-shaped  risk measures,
identified with the corresponding properties of the risk-weighting function  $g$. 
By looking at this new class, the gap between  convex risk measures and consistent risk measures, arguably opaque in the literature due to lack of explicit examples, becomes transparent. 
The empirical estimators of   $\mathrm{MD}^D_g$ can be formulated based on those of $D_h$, VaR
and the expectation, and the asymptotic normality of the estimators is established.  We find  the asymptotic variance of  $\mathrm{MD}^D_g$   is  smaller than the one of  risk measures without distortion; a useful  feature in  statistical   estimation. This intuitively illustrates a trade-off between statistical efficiency and sensitivity to risk.  

 We discuss some future directions for the research of  $\mathrm{MD}^D_g$.   In fact,   the form of $\mathrm{MD}^D_g$ (not necessarily monotonic) includes many commonly used  reinsurance premium  principle as special cases; see, e.g., the variance related principles (\cite{FL06} and \cite{C12})  and the Denneberg’s absolute deviation principle (\cite{TWWZ20}). Thus,   it would be interesting to
formulate the optimal reinsurance problem where the reinsurance principle  is computed by  $\mathrm{MD}^D_g$.  The optimal reinsurance  strategies should rely on the properties and the form of  $g$.  It is  also meaningful to consider   risk sharing  problems and portfolio selection problems  under the criterion of minimizing  $\mathrm{MD}^D_g$ under a similar  framework  to  \citet{GMZ12, GMZ13} and \cite{GZ12}. 
Another direction of generalization is to relax cash-additivity we imposed throughout the paper to cash-subadditivity, as this allows for non-constant eligible assets when computing regulatory capital requirement; see \cite{ER09} and \cite{FKM14}. 
Finally, we worked throughout with law-invariant mean-deviation measures
 with respect to a fixed probability measure. 
 When the probability measure is uncertain, one needs to develop a framework of mean-deviation measures that can incorporate uncertainty and multiple scenarios in some forms (e.g., \cite{CF17}, \cite{DKW19} and \cite{FLW23}).

\vspace{3mm}
\textbf{Acknowledgements.}  
 We thank the Editor, an Associate Editor, two anonymous referees 
for helpful comments.  Xia Han  is supported by the National Natural Science Foundation of China (Grant No.~12301604, 12371471). 
Ruodu Wang is supported by the Natural Sciences and Engineering Research Council of Canada (CRC-2022-00141, RGPIN-2024-03728) and the Sun Life Research Fellowship.


\appendix

\setcounter{table}{0}
\setcounter{figure}{0}
\setcounter{equation}{0}
\renewcommand{\thetable}{S.\arabic{table}}
\renewcommand{\thefigure}{S.\arabic{figure}}
\renewcommand{\theequation}{S.\arabic{equation}}

\setcounter{theorem}{0}
\setcounter{proposition}{0}
\renewcommand{\thetheorem}{S.\arabic{theorem}}
\renewcommand{\theproposition}{S.\arabic{proposition}}
\setcounter{lemma}{0}
\renewcommand{\thelemma}{S.\arabic{lemma}}

\setcounter{corollary}{0}
\renewcommand{\thecorollary}{S.\arabic{corollary}}

\setcounter{remark}{0}
\renewcommand{\theremark}{S.\arabic{remark}}
\setcounter{definition}{0}
\renewcommand{\thedefinition}{S.\arabic{definition}}

\setcounter{example}{0}
\renewcommand{\theexample}{S.\arabic{example}}

\section{Monotonicity of mean-deviation models}\label{app:[M]}

In this appendix, we analyze monotonicity of mean-deviation models.
Recall the necessary condition in Lemma \ref{prop:finite_K}, that is $\lambda D\in \overline{\mathcal D}^p$ for some $\lambda>0$. We aim to re-examine [M] of the mean-deviation model $U=V(\E,D)$ and establish the characterization of [M] under this condition. We assume without loss of generality that $\lambda=1$, i.e., $D\in \overline{\mathcal D}^p$. For such $D$, by definition it is
\begin{align}\label{eq-condD1}
 \sup _{X \in  {(L^p)}^\circ }\frac{D(X)}{\esssup X-\E[X]}=1.
 \end{align}

\begin{proposition}\label{prop:chMon}
Fix $p\in[1,\infty]$. Let $D\in \overline{\mathcal D}^p$ and $U=V(\E,D)$ be defined by \eqref{eq:V} with $U(X)<\infty$ for all $X\in L^p$. Suppose that either $V:\R\times [0,\infty)$ is left continuous in its second argument or the maximizer in \eqref{eq-condD1} is attainable. Then, $U$ satisfies [M] if and only if $V(m-a,d+a)\le V(m,d)$ for all $m\in\R$ and $a,d\ge 0$.
\end{proposition}

\begin{proof}

While the proof is similar to that of Proposition 4 in \cite{GMZ12}, which establishes the same necessary and sufficient condition under the assumption that $U=V(\E,D)$ is continuous, we provide the full details here for the sake of completeness and clarity.

We first verify sufficiency. For $X,Y\in L^p$ satisfying $X\le Y$ and $X\neq Y$, we denote by $a=\E[Y]-\E[X]>0$, and it holds that $\esssup(X-Y)\le 0$.
By \eqref{eq-condD1}, we have
\begin{align*}
D(X-Y)\le\esssup(X-Y)-\E[X-Y]\le a.
\end{align*}
Since $D$ satisfies (D4), we have 
\begin{align}\label{eq-strsub}
D(X)\le D(Y)+D(X-Y)\le D(Y)+a.
\end{align}
Therefore,
\begin{align*}
U(X)&=V(\E[X],D(X))=V(\E[Y]-a,D(X))\\
&\le V(\E[Y], D(X)-a)\le V(\E[Y],D(Y))=U(Y),
\end{align*}
where the first inequality follows from the assumption by letting $m=\E[Y]$, $d=D(X)-a$, and the second inequality is due to $D(X)-a\le D(Y)$ in \eqref{eq-strsub}.
Hence, we conclude that $U$ satisfies [M]. 
Conversely, we first consider the case that $V$ is left continuous in its second argument.
It follows from \eqref{eq-condD1} that for any $\epsilon>0$, there exists $X_1\in (L^p)^{\circ}$ such that 
\begin{align*}
1-\epsilon\le \frac{D(X_1)}{\esssup X_1-\E[X_1]}\le 1.
\end{align*}
Let $m\in\R$ and $d,a\ge 0$. We define
\begin{align*}
X_2=a\frac{X_1-\esssup X_1}{\esssup X_1-\E[X_1]}~~{\rm and}~~X_3=\frac{d}{a}X_2+m+d.
\end{align*}
Through standard calculation,   $\E[X_3]=m$, $(1-\epsilon)d\le D(X_3)\le d$, $\E[X_2+X_3]=m-a$ and $(a+d)(1-\epsilon)\le D(X_2+X_3)\le a+d$. Note that $X_2\le 0$ which implies $X_2+X_3\le X_3$. Using [M], we have
\begin{align*}
V(m-a,(a+d)(1-\epsilon))\le V(\E[X_2+X_3],D(X_2+X_3))\le V(\E[X_3],D(X_3))\le V(m,d).
\end{align*}
Letting $\epsilon\downarrow0$ and using the left continuity, we conclude that $V(m-a,a+d)\le V(m,d)$ for for all $m\in\R$ and $a,d\ge 0$. Now, we assume that the maximizer in \eqref{eq-condD1} is attainable, and the necessity follows a similar proof to the previous arguments by constructing $X_1$ such that $D(X_1)/(\esssup X_1-\E[X_1])=1$. Hence, we complete the proof.
\end{proof}

We note that the ES-deviation $\ES_\alpha-\E$ for $\alpha\in(0,1)$ serves as  an example where the maximizer in  \eqref{eq-condD1} is attainable.

\section{Axiomatization of monotonic mean-deviation measures}\label{app:B}

This appendix contains details on the axiomatization of  monotonic mean-deviation measures, and its connection to the results of \cite{GMZ12}. 
We first present two weaker axioms than A1 and A2, respectively. 
\begin{itemize}
\item[{B1}]    If $c_1\le c_2,$ then $c_1\succeq c_2$ for any $c_1, c_2\in \R$.
\item[B2]    For any  $X,Y$ satisfying $\E[X]=\E[Y]$ and  $c >0$,    $X \succeq Y$ if and only if $X+c \succeq  Y+c$.\end{itemize} 

\cite{GMZ12}  characterized a mean-deviation model $X\mapsto V(\E[X],D(X))$  using a set of axioms equivalent to A1, B2 and A3-A7. In contrast,
to obtain the   characterization in Theorem \ref{thm:1}, we first use   Axiom B1 to characterize $\mathrm{MD}^D_g$ in \eqref{eq:MD}, where $g$ does not necessarily satisfy the 1-Lipschitz condition, and $D$ is not necessarily weakly upper-range dominated. The characterized $\mathrm{MD}^D_g$  contains more  examples  such as the mean-variance and the mean-SD functionals, which  are not  monotonic but satisfy the weak monotonicity.   
 
\begin{proposition}\label{prop:1} 
Fix $p\in [1,\infty]$.
A preference  $\succeq$  satisfies Axioms B1 and A2--A7    if and only if  $\succeq$ can be represented by  $\mathrm{MD}^D_g=g \circ D+\E$ where    $D\in\mathcal{D}^p$ is continuous and $g:[0, \infty) \rightarrow \mathbb{R}$ is  continuous and strictly increasing.   
\end{proposition}


\begin{proof}[Proof of Proposition \ref{prop:1}]
  We first show sufficiency.  Let $\mathrm{MD}^D_g=g \circ D+\E$ represent $\succeq$  where  $D\in \mathcal D^p $  and  $g:[0, \infty) \rightarrow \mathbb{R}$ is  some continuous, non-constant and increasing function.   
  Axioms B1, A2, A3, A5 are straightforward by the properties of $D\in\mathcal D^p$.
  
  


 
 The condition that $X\le_{\rm cx} Y$ implies $X\succeq Y$ in
 Axiom A4 comes from Theorem 4.1 of \cite{D05} which showed  that   every law-invariant continuous convex measure on an atomless probability space is consistent with convex ordering. Moreover, for any $X\in (L^p)^\circ$, since $ D (X)>0$, together with the fact that $g$ is a strictly increasing function, we have $g( D (X))>g(0)$. Therefore, we have $\mathrm{MD}^D_g(X)>\mathrm{MD}^D_g(\E[X])$, which implies that  $\E[X]\succ X$ for any $X\in (L^p)^\circ$. Hence, we have verified Axiom A4.

To show Axiom A6 for $\mathrm{MD}^D_g$, for any $X,Y\in L^p$  such that $\E[X]=\E[Y]$ and $\mathrm{MD}^D_g (X)=\mathrm{MD}^D_g(Y)$, we have $g( D (X))=g( D (Y))$ and thus $ D (X)=D (Y)$ because $g$ is strictly increasing. In this case, for any $\lambda>0$, we have $$\begin{aligned}\mathrm{MD}^D_g (\lambda X+(1-\lambda) Y)&=g( D (\lambda X+(1-\lambda)Y))+\lambda \E [X]+(1-\lambda)\E[Y]\\&\leq g(\lambda  D ( X)+(1-\lambda)  D (Y))+\E[Y]\\&\leq g( D (X))+\E[X]=\mathrm{MD}^D_g (X).\end{aligned}$$   
Axiom A7 follows directly form the fact that $D$ and $g$ are continuous.

 Next, we  prove   necessity.   Axioms B1 and A2, A3 and A5 imply the existence of a unique certainty equivalence functional  $\rho: L^p  \rightarrow\R$, i.e.,   we have  $X \succeq Y \Longleftrightarrow \rho(X)\le \rho(Y)$ for any $X,Y\in L^p$, and $\rho(c)=c$ for any $c\in\R$; see  Theorem 3.3 of \cite{AB03}.    In particular, $\rho$ is continuous by Axiom A7. 
 
  Let $X_0\in (L^p)^\circ$ be such that $\E\left[X_0\right]=0$. Define $\phi(\lambda)=\rho \left(\lambda X_0\right)$ for $\lambda\ge 0$. We have $\phi(0)=\rho(0)$. The continuity of $\phi$ follows from the continuity of $\rho$.  Since $\E[\lambda X_0]\succ\lambda X_0$ for any $\lambda> 0$ by Axiom 5,  we have $\rho(0)=\rho (\E[\lambda X_0])<\rho (\lambda X_0)$.  This implies that $\phi (\lambda)>\phi(0)$ for any $\lambda>0$. Moreover, it follows from Axiom 5 that $\rho (\lambda_1 X_0)\leq\rho (\lambda_2 X_0)$  for any $0<\lambda_1<\lambda_2$ as $\lambda_1 X_0\le_{\rm cx} \lambda_2X_0.$
 Thus, we have $\phi(\lambda_1)\leq\phi(\lambda_2)$ which implies that $\phi$ is an increasing function on $[0,\infty)$.  To show the inequality is strict, we assume by the contradiction, i.e., $\phi(\lambda_1)=\phi(\lambda_2)$.  In this case, we have $\rho  (\lambda_1 X_0)=\rho (\lambda_2 X_0)$ and  $\rho  (\lambda_1 k X_0)=\rho (\lambda_2 k X_0)$ for any $k>0$ by Axiom A3. Let $k=\lambda_1/\lambda_2<1$. By induction, we have $\rho (\lambda_1 X_0)=\rho (\lambda_1 k^n X_0)$ for any $n\in \N$. Letting $n\to \infty$, by    Axiom A7, we have $\rho(\lambda_1 X_0)=\rho(0)$, which contradicts to Axiom A4.   Thus,  $\phi$ is a strictly increasing and continuous function on $[0,\infty)$,  and its inverse function  $\phi^{-1}(x):=\inf\{\lambda\in[0,\infty): \phi(\lambda)\ge x\}$ is also strictly increasing and continuous on the range of $\phi$.

 For  $X\in L^p$, let $\overline X=X-\E[X]$ and
 $ D (X)=\phi^{-1}(\rho (\overline X))$. 
 Since $\rho$  and $\phi$ are continuous functions,  we know that $D$ is continuous.    
 Next, we aim to verify that $D\in\mathcal D^p$.
 It is clear that $ D $ is   law-invariant  since $\rho$ is   law-invariant by Axiom A4, and thus (D5) holds.   
 For any $c\in\R$, $ D (X+c)=\phi^{-1}(\rho (\overline {X+c}))= D (X)$, which implies (D1). 
Note that Axiom A4 implies $\rho(\overline{X})>\rho(0)$ for all $X\in(L^p)^\circ$.  We have   $D(X)=\phi^{-1}(\rho (\overline X))>\phi^{-1}(\rho(0))=0$ as $\phi^{-1}$ is strictly increasing.
 For any $c\in\R$,  $ D (c)=\phi^{-1}(\rho (0))=0$. Thus, (D2) holds.   
For any $X\in L^p$, we have 
\begin{align}\label{eq-eq}
\rho(D (X) X_0)=\phi(D(X))=\phi\circ\phi^{-1}(\rho(\overline X))=\rho(\overline X)\Longleftrightarrow
\overline X \simeq D (X) X_0.
\end{align}
It then follows from Axiom A3 that $\lambda\overline X \simeq \lambda D (X) X_0$ for all $\lambda\ge 0$.
Hence,  we have $\rho (\overline{\lambda X})=\rho (\lambda D  (X)X_0)$. On the other hand, $\overline{\lambda X} \simeq  D (\lambda X) X_0$ implies $\rho (\overline{\lambda X})=\rho ( D (\lambda X) X_0) $. This concludes that
$\rho (\lambda D (X)X_0)=\rho ( D (\lambda X)X_0)$, which is equivalent to $\phi(\lambda D (X))=\phi( D (\lambda X))$.  Note that $\phi$ is strictly increasing. It holds that $\lambda  D (X)= D (\lambda X)$ which implies (D3).
For  $X, Y\in L^p$, if $X$ or $Y$ is constant, (D4) holds directly. Otherwise, we have $ D(X)>0$ and $ D(Y)>0$. Combining \eqref{eq-eq} and Axiom A3,
we have  $\rho (\overline X/ D (X))=\rho  (X_0)$ and   $\rho (\overline Y/ D (Y))=\rho  (X_0)$ which implies $\rho (\overline X/ D (X))=\rho (\overline Y/ D (Y))$. Moreover, by Axiom A6,  for all $\lambda\in[0,1]$,  $$\rho  \left(\lambda \frac{\overline X}{ D (X)}+(1-\lambda)\frac{ \overline Y}{ D (Y)}\right)\leq \rho  \left(\frac{\overline Y}{ D (Y)}\right)=\rho   (X_0).$$ By setting $\lambda= D (X)/( D (X)+ D (Y))$, we have   $\rho  \left( ({\overline X+ \overline Y)}/{( D (X) +  D (Y))}\right)\leq \rho  (X_0).$ 
Applying \eqref{eq-eq} and Axiom A3 again, we have the following relation:
\begin{align*}
\frac{\overline X+\overline Y}{D(X)+D(Y)}=\frac{\overline{X+Y}}{D(X)+D(Y)}
\simeq\frac{D(X+Y)X_0}{D(X)+D(Y)}.
\end{align*}
Hence, denote by $k=D(X+Y)/(D(X)+D(Y))$, and 
we have  $\rho \left (k X_0\right)\leq \rho (X_0),$  which implies  $\phi\left (k\right)\leq \phi(1).$  Noting that $\phi$ is strictly increasing, we have $D (X+Y)\leq  D (X)+ D (Y)$ and (D4) holds.

For any $X\in L^p$,
using $X\simeq \rho(X)$, we have $X-\E[X]\simeq \rho(X)-\E[X] $ by Axiom A2, which implies  
$\rho(X-\E[X])= \rho(X)-\E[X]$.
Therefore,  using \eqref{eq-eq}, 
$$
\rho(X) = \rho(X-\E[X]) + \E[X] = \rho(\overline{X})+\E[X]=
\phi(D(X))+\E[X], ~~~\mbox{for all $X\in L^p$},
$$
where the last step follows from \eqref{eq-eq}.
This completes the proof.  
\end{proof}

\begin{proof}[Proof of Theorem \ref{thm:1}] Sufficiency is straightforward by  
combining Theorem \ref{thm:monetay}, Lemma  \ref{prop:finite_K} and Proposition \ref{prop:1}.  Next, we show the    necessity. By Proposition \ref{prop:1}, $\succeq$ can be represented by $\mathrm{MD}^{D'}_{f}=f \circ D'+\E$   where    $D'\in\mathcal{D}^p$, and $f:[0, \infty) \rightarrow \mathbb{R}$ is  some continuous   and strictly increasing function.  Since $\mathrm{MD}^{D'}_{f}$ satisfies monotonicity, by Lemma \ref{prop:finite_K},  we have $D'\in \overline{\mathcal D}^p_K$. Define $g=f\circ D$ and $D=D'/K$, we have $\mathrm{MD}^{D'}_{f}=\mathrm{MD}^{D}_{g}=g \circ D+\E$ where  $g:[0, \infty) \rightarrow \mathbb{R}$ is  some continuous and strictly increasing function and $D\in \mathcal {\overline D}^p$. By Theorem \ref{thm:monetay},  $g$ is 1-Lipschitz.
Hence, we complete the proof.
\end{proof}

 \section{Proof of  Theorem \ref{thm:4}}\label{app:pf_thm5}
 This appendix contains the proof of  Theorem \ref{thm:4}.
\begin{proof}
The Law of Large Numbers yields $\widehat x_n\stackrel{\mathbb{P}}{\rightarrow} \E[X]$. 
 By Theorem 2.6 of \cite{KSZ14}, 
 the empirical estimator for a finite convex risk measure on $L^p$ is consistent, that is,
 $\widehat D (n)  + \widehat x_n \stackrel{\mathbb{P}}{\rightarrow} D(X) +\E[X]$, and this gives  $\widehat D (n)  \stackrel{\mathbb{P}}{\rightarrow} D(X)$. 
Moreover, since $g\in\mathcal G$,  we have $g(\widehat D(n))+\widehat\E[n]\stackrel{\mathbb{P}}{\rightarrow}g(D(X))+\E[X]$.
This proves the first part of the result.


Next, we will show the asymptotic normality.
Let $B=(B_t)_{t\in [0,1]}$ be a standard Brownian bridge, and let $d_n= \sqrt{n} (\widehat{D}(n)-D(X) ) $,
  $e_n =\sqrt{n} (\widehat x_n-\E[X])  $, 
  and $g_n= \sqrt{n} (g(\widehat{D}(n))-g(D(X) )) $. 
We need to first show  
\begin{align}
    \label{eq:joint-conv}
(d_n,e_n)   \stackrel{\mathrm{d}} {\rightarrow}(Z,W):=\left(  \int_0^1 \frac{B_s h'(1-s)}{\tilde f(s)} \mathrm{d} s, \int_0^1 \frac{B_s}{\tilde f(s)} \mathrm{d} s\right). 
\end{align} 
By the Cram\'er-Wold theorem, it suffices to show 
\begin{align}
    \label{eq:joint-conv2}
 a d_n+ be_n   \stackrel{\mathrm{d}}\rightarrow  a Z+ b W \mbox{~for all $a,b\in \R$.} 
\end{align} 
Note that $aD + b \E$
can be written as an integral of the quantile, that is, 
$$
a D(X) + b\E[X] = \int_0^1 F^{-1}(t) (a h'(1-t)+b)\d t.
$$
Denote by $A_n$ 
the empirical quantile process, that is,
$$
A_n(t)= \sqrt{n} (\widehat F^{-1}_n(t) - F^{-1}(t)), ~~~t\in (0,1).
$$
It follows that 
$$ a d_n+ be_n = 
\int_0^1 A_n(t) (a h'(1-t)+b)\d t.
$$
Using this representation, the convergence \eqref{eq:joint-conv2} can be verified using Theorem 3.2 of \cite{JZ03}, and we briefly verify it here.
It is well known that, under Assumption \ref{assump:1}, as $n\to \infty$, $A_n$
converges to the Gaussian process $B/\tilde f$   in $L^\infty[1-\delta,1+\delta]$ for any $\delta>0$ (see e.g.,  \cite{DGU05}).
This yields 
 $$
\int_\delta^{1-\delta} A_n(t) (a h'(1-t)+b)\d t \stackrel{\mathrm{d}}\rightarrow  
\int_\delta^{1-\delta}  \frac{B_t}{\tilde f(t)}(a h'(1-t)+b)\d t.
$$ 
To show \eqref{eq:joint-conv2}, it suffices to verify      \begin{equation} 
\label{eq:joint-conv3}
\int_\delta^{1-\delta}  \frac{B_t}{\tilde f(t)}(a h'(1-t)+b)\d t\to  \int_0^1 \frac{B_t}{\tilde f(t)}(a h'(1-t)+b)\d t \mbox{~~~as $\delta \downarrow 0$}.
\end{equation}
Denote by $w_t=t(1-t)$. 
Since $h\in \Phi^p$ and 
$X\in L^{\gamma}$, we have,  for some $C>0$, 
$$|h'(1-t) |\le C   w_t^{1/p-1};~~~
|F^{-1}(t)| \le C w_t^{-1/\gamma};~~~ \frac{1}{\tilde f(t) } =
\frac{\d F^{-1}(t)} {\d t} \le C  w_t ^{-1/\gamma -1}. 
$$ 
Note that $1/p-1/\gamma >1/2$ and $B_t=o_{\p}(w_t^{1/2-\epsilon})$ for any $\epsilon>0$ as $t\to 0$ or $1$. Hence,  
for some $\eta>0$,
$$
\left|B_t \frac{a h'(1-t)+b}{\tilde f(t)}\right| =o_{\p} (w_t^{\eta-1}) \mbox{~~~for $t\in (0,1)$},
$$
and this 
guarantees \eqref{eq:joint-conv3}. Therefore, \eqref{eq:joint-conv} holds. 

By the  Mean Value Theorem, there exists $x_0$ between $D(X)$ and $\widehat D(n)$  such that 
$$ 
\sqrt{n}(g (\widehat D(n))-g(D(X)))=g'(x_0)\sqrt{n}(\widehat D(n)-D(X)). $$
Using the fact that $\widehat D(n)\stackrel{\mathbb{P}}{\rightarrow}  D(X)$,  we get  
$$
(g_n,e_n)   \stackrel{\mathrm{d}}{\rightarrow}\left(   g'(D(X))\int_0^1 \frac{B_s h'(1-s)}{\tilde f(s)} \mathrm{d} s , \int_0^1 \frac{B_s}{\tilde f(s)} \mathrm{d} s\right). $$ 
Hence,  $$
\begin{aligned}
\sqrt{n}\left(\widehat{\mathrm{MD}}^D_g(n)-\mathrm{MD}^D_g(X)\right) =g_n+e_n\stackrel{\mathrm{d}}{\rightarrow}g'(D(X))\int_0^1 \frac{B_sh'(1- s)}{\tilde f(s)} \mathrm{d} s+ \int_0^1 \frac{B_s}{\tilde f(s)} \mathrm{d} s,
\end{aligned}
$$
or equivalently,
 $$
\begin{aligned}
\sqrt{n}\left(\widehat{\mathrm{MD}}^D_g(n)-\mathrm{MD}^D_g(X)\right) \stackrel{\mathrm{d}}{\rightarrow}\int_0^1 \frac{B_s}{\tilde f(s)}(h'(1- s)g'(D(X))+1)  \mathrm{d} s.
\end{aligned}
$$
Using the convariance property of the Brownian bridge, that is,
$\mathrm{Cov}(B_t, B_s)=s-s t$ for $s<t$, we have
\begin{align*}&\var\left[\int_0^1 \frac{B_s(h'(1- s)g'(D(X))+1)}{\tilde f(s)} \mathrm{d} s\right]\\ & =\mathbb{E}\left[\int_0^1 \int_0^1 \frac{(h'(1- s)g'(D(X)) +1)(h'(1- t)g'(D(X))+1)B_s B_t}{\tilde f(s) \tilde f(t)} \mathrm{d} t \mathrm{d} s\right]\\ & = \int_0^1 \int_0^1 \frac{(h'(1- s)g'(D(X)) +1)(h'(1- t)g'(D(X))+1)(s \wedge t-s t)}{\tilde f(s) \tilde f(t)} \mathrm{d} t \mathrm{d} s. \end{align*} 
Thus,  $
\sqrt{n}\left(\widehat{\mathrm{MD}}^D_g(n)-\mathrm{MD}^D_g(X)\right) \stackrel{\mathrm{d}}{\rightarrow} \mathrm{N}(0,\sigma^2_g)
$ in which $\sigma^2_g$ is given by \eqref{eq:sigma_g}.
\end{proof}

\section{Worst-case values under model uncertainty}\label{sec:5}
In the context of robust risk evaluation, one may only have partial information on a risk $X$ to be evaluated. Thus, we  further discuss two model uncertainty problems based on  $\mathrm{MD}^D_g$.    

We  first consider the case in which one only knows the mean and the variance of $X$.  This setup has wide applications in model uncertainty and portfolio optimization.
Denote by $L^2(m, v)=\left\{X \in L^2: \mathbb{E}[X]=m, ~\sigma^2(X)=v^2\right\}$.   For a fixed $g\in\mathcal G$ and $D\in \mD^p$ with $p\in[1,2]$, we  consider the following worst-case problem 
\begin{equation}\label{eq:rb1}
\overline{\mathrm{MD}}^D_g(m, v)=\sup \left\{\mathrm{MD}^D_g(X): X \in L^2(m, v)\right\}.
\end{equation}
{\color{black} 
Since $g$ is increasing, the way to solve \eqref{eq:rb1} is similar to those used for related worst-case risk measures studied in the literature; see, for example, \cite{LCLW20}, \cite{PWW20} and \cite{BPV24}. In particular, if  $g(x) = \lambda x$ for some $0 < \lambda \leq 1$, \eqref{eq:rb1} reduces to the problem studied in Section 5 of \cite{PWW20}.} 
\begin{proposition}\label{prop:3}
Suppose that $p\in[1,2]$, $m \in \mathbb{R}$, $v>0$, $g\in\mathcal G$ and  $D\in\mathcal D^p$ is in \eqref{eq:h}.  We have $$\overline{\mathrm{MD}}^D_g(m, v)=\sup _{h \in \Psi^p}g\left(v \left\|h^{\prime}\right\|_2\right)+m.$$ 
\end{proposition}

\begin{proof} By \eqref{eq:h} and \eqref{eq:D_h}, we have 
\begin{align}\label{eq-MVrobust1}
\overline{\mathrm{MD}}^D_g(m, v)&=\sup_{X\in L^2(m,v)}g(D(X))+m
\notag \\&=\sup_{X\in L^2(m,v)} g\left(\sup _{h \in \Psi^p}\left\{\int_0^1 \VaR_\alpha(X) h'(1-\alpha)\mathrm{d} \alpha\right\} \right) +m\notag\\
&=\sup_{X\in L^2(m,v)} g\left(\sup_{h\in\Psi^p} D_h(X)\right)+m
= g\left(\sup_{h\in\Psi^p} \sup_{X\in L^2(m,v)} D_h(X)\right)+m,
\end{align}
where the last step holds because $g$ is increasing.
By Theorem 3.1 of \cite{LCLW20}, we have that 
$
\sup_{X\in L^2(m,v)} D_h(X)=v\left\|h^{\prime}\right\|_2
$
for any $h \in \Phi^p$. This completes the proof.
  \end{proof} 
  
\begin{remark}The worst-case problem formulated in \eqref{eq:rb1} can be extended to  the case  of other central moment instead of the variance. For $a>1, m \in \mathbb{R}$ and $v>0$, denote by \begin{equation}\label{eq:La}L^a(m, v)=\left\{X \in L^a: \mathbb{E}[X]=m, ~\mathbb{E}\left[|X-m|^a\right]=v^a\right\}.\end{equation}
Suppose that $p\in[1,a]$.
 Theorem 5 of \cite{PWW20} implies  that 
 $$\sup \left\{D_h(X): X \in L^p(m, v)\right\}=v[h]_q,~~~h\in\Phi_p,
$$
 where $q=(1-1/p)^{-1}$, $D_h$ is defined by \eqref{eq:D_h} and
$[h]_q=\min_{x\in\R}\|h'-x\|_q$.
Therefore, for $D\in\mathcal D^p$ defined by \eqref{eq:h}, it follows the similar arguments in the proof of Proposition \ref{prop:3} that
$$
\left\{{\rm MD}_g^D(X):X\in L^a(m,v)\right\}=\sup _{h \in \Psi^p}g\left(v[h]_q\right)+m.
$$
\end{remark}

\begin{example}Let $D=\ES_\alpha-\E$ with $\alpha\in(0,1)$. We have $D=D_h$, where $h(t)=(t-\alpha)_+/(1-\alpha)-t$ for $t\in[0,1]$. It holds that
\begin{align*}
[h]_q=\min_{x\in\R}\|h'-x\|_q
=\min_{x\in\R} \left(\alpha|1+x|^q+(1-\alpha)\left|\frac{\alpha}{1-\alpha}-x\right|^q\right)^{1/q}.
\end{align*}
By standard manipulation, we conclude that the minimizer of the above optimization problem can be attained at $x^*=(\alpha(1-\alpha)^{p-2}-\alpha^{p-1})/(\alpha^{p-1}+(1-\alpha)^{p-1})$, and the optimal value is $[h]_q=\alpha\left(\alpha^p(1-\alpha)+\alpha(1-\alpha)^p \right)^{-1 / p}$.
Thus, in this case, we have
$$
\overline{\mathrm{MD}}^D_g(m, v)=m + g\left(v  \alpha\left(\alpha^p(1-\alpha)+\alpha(1-\alpha)^p \right)^{-1 / p}\right).
$$ 
We compare the results for normal, {Pareto} and exponential distributions  with  the worst-case distribution with the same mean and variance. Setting $p=2$ and both mean and variance to $1$, we show the values of ${\mathrm{MD}}^D_g$ and $\overline {\mathrm{MD}}^D_{g}$  when $D=\ES_\alpha-\E$ for different  values of $\alpha\in[0.9, 0.99]$ in   Figure \ref{fig:worst_mv}. {\color{black}In particular, when $g(x) = x$, $\mathrm{MD}^D_g$ simplifies to $\mathrm{ES}_\alpha$. Given that $g$ is 1-Lipschitz, it is expected that the worst-case values of $\mathrm{ES}_\alpha$ will be larger than those of $\overline{\mathrm{MD}}^D_g$ for other forms of $g$.}  \begin{figure}[htb!]
\centering
 \includegraphics[width=17cm]{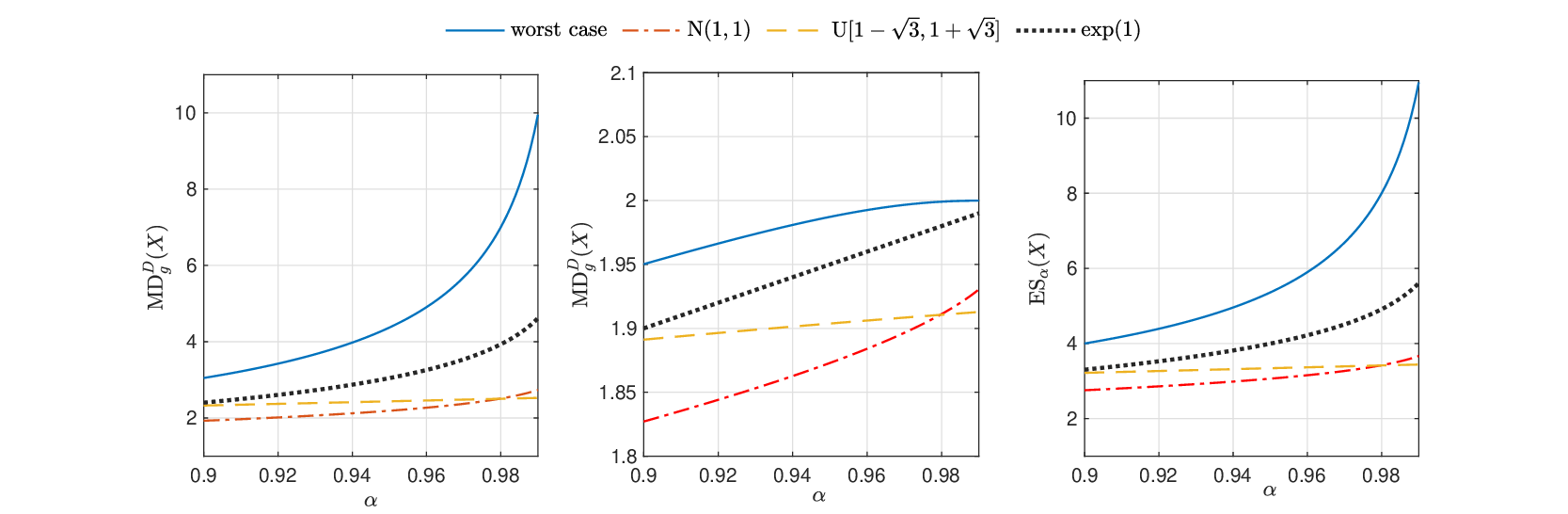}
 \captionsetup{font=small}
{\caption{   \small The values  of ${\mathrm{MD}}^D_g$ and $\overline {\mathrm{MD}}^D_{g}$ with   $g(x)=x+e^{-x}-1$  (left), $g(x)=1-e^{-x}$ (middle) and $g(x)=x$ (right)}\label{fig:worst_mv}}
\end{figure}

\end{example}

Optimization problems under the uncertainty set of a Wasserstein ball are also common in the literature when quantifying the discrepancy between a benchmark distribution and alternative scenarios; see e.g., \cite{EK18}.  
For two distributions $F$ and $G$, the type-$p$ Wasserstein metric with $p \geqslant 1$,  is given by$$
W_p(F, G)=\left(\int_0^1\left|F^{-1}(u)-G^{-1}(u)\right|^p \mathrm{~d} u\right)^{1 / p}.
$$
Denote by $\mathcal M_p$ the set of all distribution functions that have finite $p$th moment. For $F_0\in\mathcal M_p$ and $\epsilon\ge0$, we define the following uncertainty set based on the type-$p$ Wasserstein metric
\begin{align*}
\mathcal B_p(F_0,\epsilon)=\{F\in \mathcal M_p: W_p(F,F_0)\le \epsilon\}.
\end{align*}
The above uncertainty set is also known as a type-$p$ Wasserstein ball (see e.g., \cite{K19}, \cite{WLM22} and \cite{BPV24}), where $F_0$ is the center and $\epsilon$ is the radius. Note that  $\epsilon=0$ corresponds to the case of no model uncertainty.
In what follows, we focus on the type-$2$ Wasserstein ball.
For any $\varepsilon \geqslant 0$, $g\in\mathcal G$ and $D\in  \mD^p$ with $p\in[1,2]$, we define the worst-case   ${\mathrm{MD}}^D_g$    as
 $$\widetilde{\mathrm{MD}}^D_g
(X|\epsilon)=\sup \left\{\mathrm {MD}^D_g(Y): F_Y\in \mathcal B_2(F_X,\epsilon)
\right\}.
$$  
 The following proposition
gives a formula to compute the worst-case value of $\mathrm{MD}^D_g$ under the type-2 Wasserstein ball.

\begin{proposition}\label{prop:4}
Suppose $p\in[1,2]$, $g\in\mathcal G$ and $D=D_h$  in  \eqref{eq:D_h} where $h\in\Phi^p$.  
We  have 
\begin{align*}
\widetilde{\mathrm{MD}}^D_g
(X|\epsilon)=\sup_{t\in[-1, 1]}\left\{g\left(\epsilon\sqrt{1-t^2}  \left\|h^{\prime}\right\|_2
+D(X)\right)+t\epsilon+\E[X]\right\},  {~~~\epsilon\ge 0, ~X\in L^2.}
\end{align*}
\end{proposition}

\begin{proof}  
Denote by $\mathcal M=\{F_Y: Y\in L^2,~\|Y-X\|_2\le \epsilon\}$.  We first aim to show that $\mathcal M=\mathcal B_2(F_X,\epsilon)$. Note that $\mathcal B_2(F_X,\epsilon)=\{F_Y: Y\in L^2,~\int_0^1 |F_Y^{-1}(u)-F_X^{-1}(u)|^2\d u\le \epsilon^2\}$. It is obvious that $\mathcal M\subseteq \mathcal B_2(F_X,\epsilon)$ since $\|Y-X\|_2^2\ge \int_0^1 |F_Y^{-1}(u)-F_X^{-1}(u)|^2\d u$ for any $X,Y\in L^2$. To see the converse direction, for any $F\in \mathcal B_2(F_X,\epsilon)$, let $Y\in L^2$ be such that $Y$ and $X$ are comonotonic and $Y$ has distribution $F$. It holds that $\|Y-X\|_2^2=\int_0^1 |F^{-1}(u)-F_X^{-1}(u)|^2\d u\le \epsilon^2$, where the last step is due to $F\in \mathcal B_2(F_X,\epsilon)$. Hence, we have $F\in \mathcal M$. This implies that $\mathcal B_2(F_X,\epsilon)\subseteq\mathcal M$, and we have concluded that $\mathcal M=\mathcal B_2(F_X,\epsilon)$. 
Note that ${\rm MD}_g^D$ is law-invariant.
We have
\begin{align}\label{eq-RW1}
\widetilde{\mathrm{MD}}^D_g(X|\epsilon)
 &=\sup\{{\rm MD}_g^D(Y): F_Y\in \mathcal B_2(F_X,\epsilon)\}\notag\\
&=\sup_{\|Y-X\|_2\le \epsilon} {\rm MD}_g^D(Y)
=\sup _{ \|Y-X\|_2 \leq \epsilon} \left\{g(D(Y))+\E[Y]\right\}.
\end{align}
Denote by $\mu_0=\E[X]$. It holds that 
$$
\{\E[Y]: \|Y-X\|_2\le \epsilon\}=\{\mu_0+\E[V]: \|V\|_2\le \epsilon\}\subseteq[\mu_0-\epsilon,\mu_0+\epsilon].
$$
Therefore, \eqref{eq-RW1} reduces to
$$\begin{aligned}\sup _{ \|Y-X\|_2 \leq \epsilon} \left\{g(D(Y))+\E[Y]\right\}
&=\sup_{\mu\in[\mu_0-\epsilon,\mu_0+\epsilon]}~~~~\sup _{ \|Y-X\|_2 \leq \epsilon,~ \E[Y]=\mu} \left\{g(D(Y))+\E[Y]\right\}\\
&=\sup_{\mu\in[\mu_0-\epsilon, \mu_0+\epsilon]}~~~~\sup _{ \|V\|_2 \leq \epsilon,~ \E[V]=\mu-\mu_0} \left\{g(D(V+X))+\mu\right\}\\&=\sup_{\mu\in[\mu_0-\epsilon,\mu_0+\epsilon]}~~~~\sup _{ \|V\|_2 \leq \epsilon,  ~\E[V]=\mu-\mu_0} \left\{g(D(V)+D(X))+\mu\right\}\\&= \sup_{\mu\in[\mu_0-\epsilon,\mu_0+\epsilon]}~~~~\sup _{ 
\sigma^2(V)\leq \epsilon^2-(\mu-\mu_0)^2,  ~\E[V]=\mu-\mu_0}\left\{g(D(V)+D(X))+\mu\right\},\end{aligned}$$ 
where the third equality holds because $g$ is increasing and $D$ defined in \eqref{eq:D_h}  is subadditive and comonotonic additive, and  we can construct $V$ and $X$ to be comonotonic. Since $g$ is increasing, the inner optimization problem is equivalent to maximizing  $D(V)$ over $\{V: \sigma^2(V) \leq \epsilon^2-(\mu-\mu_0)^2,  ~\E[V]=\mu-\mu_0\}$. Using the arguments in the proof of Proposition \ref{prop:3}, we have 
$$\begin{aligned}
\sup\{D(V): \sigma^2(V)\leq \epsilon^2-(\mu-\mu_0)^2,  ~\E[V]=\mu-\mu_0\}=\sqrt{\epsilon^2-(\mu-\mu_0)^2}  \left\|h^{\prime}\right\|_2.
\end{aligned}$$  Therefore, we have  $$\begin{aligned}\widetilde{\mathrm{MD}}^D_g
(X|\epsilon)
&=\sup_{\mu\in[\mu_0-\epsilon,\mu_0+\epsilon]}\left\{g\left(\sqrt{\epsilon^2-(\mu-\mu_0)^2}\|h'\|_2+D(X)\right)+\mu\right\}\\
&=\sup_{t\in[-\epsilon, \epsilon]} \left\{g\left(\sqrt{\epsilon^2-t^2} \left\|h^{\prime}\right\|_2
+D(X)\right)+t+\mu_0 \right\}\\& =\sup_{t\in[-1, 1]} \left\{g\left(\epsilon\sqrt{1-t^2} \left\|h^{\prime}\right\|_2
+D(X)\right)+t\epsilon+\E[X]\right\},\end{aligned}$$
which completes the proof.
\end{proof}

In Proposition \ref{prop:4},   our  analysis is confined to the case of type-$2$ Wasserstein ball and signed Choquet integral $D_h$.
Working with general deviation measure  $D$ is not more difficult, as it only involves another supremum over $\Psi^p$ by using  \eqref{eq:h}.
For the general type-$p$ Wasserstein ball with $p\ne 2$, following similar arguments to those used in the proof of Proposition \ref{prop:4} leads us to 
\begin{align}\label{eq-eqgp}
\sup \left\{\mathrm {MD}^D_g(Y): F_Y\in \mathcal B_p(F_X,\epsilon)\right\}
=\sup_{\mu\in[\E[X]-\epsilon, \E[X]+\epsilon]}~~~~\sup _{ \|V\|_p \leq \epsilon,~ \E[V]=\mu-\E[X]} \left\{g(D(V+X))+\mu\right\}.
\end{align}
We do not have an explicit formula to solve the inner supremum problem in the right-hand side of \eqref{eq-eqgp}. This is due to the fact that $\Vert V \Vert_p$ and $\E[V]$ does not align very well unless $p=2$.

\begin{example}  
Let $g(x)=x-\log(1+x)$ for $x\in\R$ and $h(t)=(t-\alpha)_+/(1-\alpha)-t$ for $t\in[0,1]$ with $\alpha\in(0,1)$.
We have $D:=D_h=\ES_\alpha-\E$, and it follows from Proposition \ref{prop:4} that
\begin{align*}&\widetilde{\mathrm{MD}}^D_g
(X|\epsilon) =\sup_{t\in[-1, 1]}\left\{g\left(\epsilon\sqrt{1-t^2}\sqrt{\frac{\alpha}{1-\alpha}}+\ES_\alpha(X)-\E[X]\right)+t\epsilon+\E[X]\right\}\\& =\sup_{t\in[-1, 1]}\left\{t\epsilon+\epsilon\sqrt{1-t^2}\sqrt{\frac{\alpha}{1-\alpha}}+\ES_\alpha(X)-\log\left(1+\epsilon\sqrt{1-t^2}\sqrt{\frac{\alpha}{1-\alpha}}+\ES_\alpha(X)-\E[X]\right)\right\}.\end{align*}

Next, let $D={\rm Gini}$ defined in \eqref{eq:Gini-d}.
By Proposition \ref{prop:4}, we have  $$\begin{aligned}\widetilde{\mathrm{MD}}^D_g
(X|\epsilon)&=\sup_{t\in[-1, 1]}\left\{g\left(\frac {\sqrt{3}\epsilon} {3}\sqrt{1-t^2}+\mathrm{Gini}(X)\right)+t\epsilon+\E[X]\right\}\\&=\sup_{t\in[-1, 1]} \left\{t\epsilon+\E[X]+\frac {\sqrt{3}\epsilon} {3}\sqrt{1-t^2}+\mathrm{Gini}(X)-\log\left(1+\frac {\sqrt{3}\epsilon} {3}\sqrt{1-t^2}+\mathrm{Gini}(X)\right) \right\}. \end{aligned}
$$  
{\color{black}The maximum values are computed numerically by considering  $g(x)=x-\log(1+x)$ and $g(x)=x$ with the benchmark distributions being  normal,
Pareto and exponential.  We calculate  the worst values of $\mathrm{MD}^D_g(X)$ when $D=\ES_\alpha-\E$ with $\alpha=0.9$ and $D=\mathrm{Gini}$, across  different values of uncertainty level $\epsilon$.   Results are presented  in    Figures \ref{fig:worst_wass1} and \ref{fig:worst_wass2}, respectively.   Again, when $g(x) = x$, $\mathrm{MD}^D_g$ simplifies to $\mathrm{ES}_\alpha$ or $\mathrm{Gini}+\E$. Given that $g$ is 1-Lipschitz,  the worst-case values of $\mathrm{ES}_\alpha$ or  $\mathrm{Gini}+\E$ will be larger than those of $\widetilde{\mathrm{MD}}^D_g$ for other forms of $g$. }

\begin{figure}[h]
    \centering
    \begin{minipage}{0.6\textwidth}
        \centering
        \includegraphics[width=\textwidth]{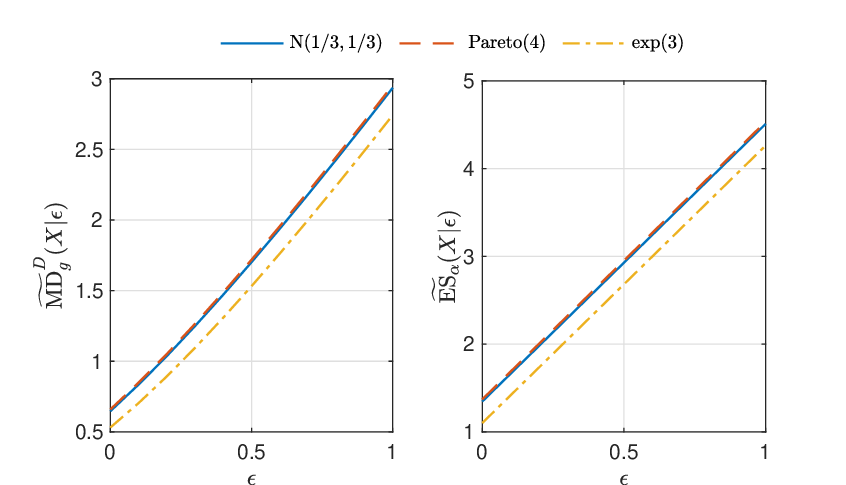}        \caption{The values of  $\widetilde{\mathrm{MD}}^D_g
(X|\epsilon)$ with   $g(x)=x-\log(1+x)$ (left) and  $g(x)=x$ (right), and $D=\ES_\alpha-\E$ }
        \label{fig:worst_wass1}
    \end{minipage}
  \hspace{0.01\textwidth}
    \begin{minipage}{0.6\textwidth}
        \centering
        \includegraphics[width=\textwidth]{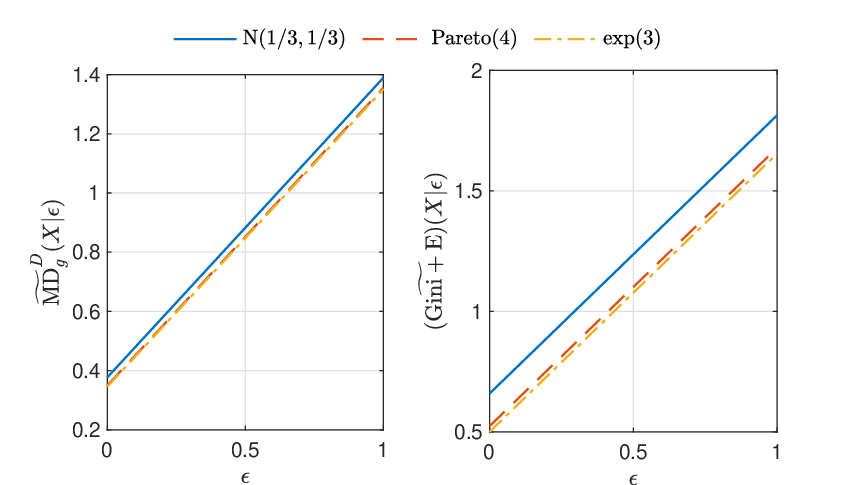} 
        \caption{The values of  $\widetilde{\mathrm{MD}}^D_g
(X|\epsilon)$  with   $g(x)=x-\log(1+x)$ (left) and  $g(x)=x$ (right), and $D=\mathrm{Gini}$}
        \label{fig:worst_wass2}
    \end{minipage}
\end{figure}
\end{example}


\end{document}